\documentclass[11pt]{article}
\usepackage[margin=1.15in]{geometry}
\usepackage{amsmath,amssymb,amsthm}
\usepackage{booktabs}
\usepackage{graphicx}
\usepackage{enumitem}
\graphicspath{{../}{./}}
\usepackage[colorlinks=true,linkcolor=blue,citecolor=blue,urlcolor=blue]{hyperref}

\newtheorem{theorem}{Theorem}

\newtheorem{lemma}{Lemma}
\newtheorem{proposition}[lemma]{Proposition}
\newtheorem{corollary}[lemma]{Corollary}
\newtheorem{conjecture}{Conjecture}
\theoremstyle{definition}
\newtheorem{definition}{Definition}
\theoremstyle{remark}
\newtheorem{remark}{Remark}

\newcommand{\OO}{\Omega}
\newcommand{\RR}{\mathbb{R}}

\title{One Shot, Twenty-One Balls:\\
A Computer-Assisted Proof of a Single-Stroke\\
Total Clearance in an Idealized Model of Snooker}
\author{Avner Kantor\\[2pt]
\small Israel Central Bureau of Statistics}
\date{}

\begin{document}
\maketitle

\begin{abstract}
We consider an explicit model of billiard dynamics on a snooker table.
Twenty-two equal disks move on a rectangle with six pocket points,
under constant rolling deceleration, frictionless ball-ball collisions
with normal restitution, and lossy specular cushions. We exhibit a
configuration of the cue ball and the twenty-one object balls, and
prove by a computer-assisted argument in interval arithmetic that every
stroke in an explicit box of side $2 \cdot 10^{-90}$ around a fixed
stroke pockets all twenty-one object balls. The set of such strokes
therefore has positive Lebesgue measure. In this configuration the
geometry of every trajectory is independent of the speed of the stroke,
and we use this to prove that every stroke with speed in an interval
of length $2.4$ m/s and direction within $10^{-90}$ radians of a fixed
direction also pockets every ball. The proofs reduce to finitely many
explicit interval inequalities; those of the box have been checked by
two independent implementations. We then determine the window of directions around the fixed stroke
for which its list of events is realized: with a first-order
enclosure in the direction variable, we prove that it has angular
width $9.481 \times 10^{-43}$ radians to four significant figures,
for every speed in the interval, and a first-order analysis of the
chain of near-grazing collisions accounts for the number. The speeds are also determined: along a fixed skeleton the squared speed of every ball is an affine function of the squared speed of the stroke, so the ribbon ends where its slowest ball comes to rest at the lip of its pocket. We locate this edge to within $1.1 \times 10^{-9}$ m/s along the fixed direction, and certify the window of directions for all speeds from $9.3214$ m/s to $12$ m/s. A second configuration, a cluster in which the cue ball strikes twelve
balls directly and the other nine are struck by object balls struck before them, has a
certified window of directions at least $9.4 \times 10^{-4}$ radians
wide, thirty-nine orders of magnitude wider; its certificate carries the
geometry in Taylor models of degree six, because first-order enclosures
inflate along near-frozen collisions. For the regulation opening position we state
the corresponding conjecture, and show that if it holds with a regular
witness then a finite certified search establishes it. No procedure is
known for the converse.
\end{abstract}

\section{Introduction}

Snooker folklore holds that no single stroke can pocket all
twenty-one object balls. This paper examines the claim as a statement
about an explicit dynamical system, Model M (Section~\ref{sec:model}),
and shows that it is false there, robustly. A companion expository
article discusses the folklore, the physics, and the statistics of
random strokes. Here we give the precise statements and proofs.

Our main results are the following. Let $c^{*}$ and
$x^{*} = (V_0^{*}, \varphi^{*})$ be the configuration and stroke of
Appendix~\ref{app:witness}.

\begin{itemize}[leftmargin=2em]
\item Theorem~\ref{thm:box}. Every stroke within $10^{-90}$ of $x^{*}$
in both speed and direction pockets all twenty-one object balls,
through the same sequence of 61 transversal events.
\item Theorem~\ref{thm:ribbon}. Every stroke with speed in
$[V_0^{*} - 2.3, V_{\max}]$, an interval of length about $2.4$ m/s, and
direction within $10^{-90}$ of $\varphi^{*}$ pockets all twenty-one
object balls.
\item Theorem~\ref{thm:width}. The set of strokes that realize the
witness's list of events contains the rectangle of ribbon speeds
times a window of directions of width $9.481 \times 10^{-43}$, and
no stroke just outside that window, at any speed. The window is
determined to four significant figures.
\item Theorem~\ref{thm:edge}. The ribbon ends in speed where ball 19
comes to rest at the lip of its pocket. For every direction $\varphi$
in the window, no stroke with speed at most a threshold $V_c(\varphi)$,
which lies between $9.3203$ and $9.3214$ m/s, realizes the list of
events, and every stroke with speed in $[9.321315, V_{\max}]$ does. At
$\varphi^{*}$ the threshold is $9.3212701684$ m/s and the edge is
located to within $1.1 \times 10^{-9}$ m/s.
\item Theorem~\ref{thm:cluster}. A second configuration, a cluster in
which the cue ball strikes twelve balls directly and the other nine are
struck by object balls struck before them, four of them across a gap
of $1.5$ mm, has a window of
directions at least $9.4 \times 10^{-4}$ radians wide: every stroke in
an explicit rectangle of that width in direction and $10^{-2}$ m/s in
speed pockets all twenty-one object balls.
\item Proposition~\ref{prop:search}. If the regulation opening position
admits a regular stroke that pockets all twenty-one balls, a finite
search with interval certification finds and proves one.
\end{itemize}

Theorems~\ref{thm:box}, \ref{thm:ribbon}, \ref{thm:width},
\ref{thm:edge} and~\ref{thm:cluster} are computer-assisted. Their proofs reduce,
through Lemma~\ref{lem:step}, Proposition~\ref{prop:skeleton} and
Lemmas~\ref{lem:scaling}, \ref{lem:skeleton}, \ref{lem:centered}
and~\ref{lem:taylor},
to finitely many explicit inequalities between interval enclosures,
checked in rigorous interval arithmetic \cite{moore2009,tucker2011}.
The inequalities of Theorem~\ref{thm:box} have been checked by two
independent implementations, those of the other four theorems by one;
the statements of those four have also been checked, at 52 single
strokes, with the second implementation (Appendix~\ref{sec:repro}).

On the side of method, the certificates form a ladder of enclosures.
Plain interval boxes certify Theorems~\ref{thm:box}
and~\ref{thm:ribbon}, but only at half-width $10^{-90}$ in direction;
at $10^{-88}$ they are undecided, although the window is about
$10^{-42}$ wide. First-order centered forms in the direction certify
pieces up to $2.4 \times 10^{-43}$ wide and cover the window with
twenty-three pieces, where boxes would need some $5 \times 10^{47}$
(Section~\ref{sec:thin}). On the cluster the centered forms fail in
turn and would need pieces narrower than $4 \times 10^{-12}$, while
Taylor models of degree six certify pieces $1.5 \times 10^{-5}$ wide
(Section~\ref{sec:cluster}). In each case the enclosure of lower order
loses far more to the dependency problem of interval arithmetic than
the dynamics amplifies. Table~\ref{tab:certs} in
Appendix~\ref{sec:repro} lists the size of every certificate.

The mathematical content is otherwise elementary. Between events the
motion is explicit, a finite composition of continuous maps is
continuous, and a nonempty open set has positive measure.
Section~\ref{sec:clocks} also records the structure behind these
statements: the strokes with a given list of events form a
semialgebraic set, so the set of successful strokes is a countable
union of semialgebraic sets, and a finite one if the number of events
is bounded, which we do not know.
The interest lies in the construction, in its certification, and in the
striking shape of the set of successful strokes, which is a long,
thin ribbon. Section~\ref{sec:thin} determines its width and explains
the number by a first-order analysis of near-grazing collisions, and
Section~\ref{sec:edge} determines where it ends in speed.
Section~\ref{sec:break} states the corresponding conjecture for the
opening position.
Section~\ref{sec:cluster} shows that the thinness belongs to the chain
of thin cuts and not to the problem, Section~\ref{sec:search} records
how the two configurations were found, and Section~\ref{sec:open}
collects open problems. Rigorous statements
about idealized billiards have a long tradition
\cite{sinai1970,tabachnikov2005}, computer-assisted proofs among them
\cite{schwartz2009}. This paper points that tradition at a
bar bet.

\section{The model}
\label{sec:model}

\begin{definition}[Table, balls, pockets]
\label{def:table}
The table is the closed rectangle $T = [0, L] \times [0, W]$ with
$L = 3.569$ and $W = 1.778$ (meters). Balls are closed disks of common
radius $R = 0.02625$. There are six pocket points $P_1, \dots, P_6$ on
the boundary, the four corners and the two midpoints of the long
sides. A ball is \emph{pocketed}, and removed from the table, as its
center enters the open disk of radius $\rho = 0.055$ about a pocket
point: at the earliest instant $t$ such that the center lies in that
disk at all times shortly after $t$.
\end{definition}

\begin{definition}[Dynamics]
\label{def:dynamics}
Between events, a moving ball with velocity $v \ne 0$ decelerates at
the constant rate $a = 0.10$ m/s$^2$ antiparallel to $v$. Its center
follows $p + v t - \tfrac{a}{2} t^{2} v / |v|$ until it stops at time
$|v| / a$, after which it is at rest. A \emph{ball-ball collision}
occurs when two centers reach distance $2R$. It is instantaneous,
frictionless and equal-mass. With $n$ the unit vector along the line
of centers and $u_1, u_2$ the normal components of the velocities,
\[
u_1' = \tfrac{1-e}{2} u_1 + \tfrac{1+e}{2} u_2, \qquad
u_2' = \tfrac{1+e}{2} u_1 + \tfrac{1-e}{2} u_2, \qquad e = 0.95,
\]
and tangential components are unchanged. A \emph{cushion contact}
occurs when a center reaches one of the lines $x = R$, $x = L - R$,
$y = R$, $y = W - R$ while moving toward it. The velocity component
normal to that line is reversed and multiplied by $e_c = 0.90$. If a
center reaches two such lines at the same instant, both components are
treated in this way. Events that occur at the same instant and have no
ball in common are applied independently. If a ball takes part in two
events at the same instant, apart from the case just described, the
evolution is undefined.
\end{definition}

\begin{definition}[Configurations, strokes, outcome]
\label{def:shots}
A \emph{configuration} $c$ places the cue ball and twenty-one object
balls at pairwise distances greater than $2R$, with centers in
$(R, L-R) \times (R, W-R)$ and at distance greater than $\rho$ from
every pocket point. We call such
$c$ \emph{admissible}. A \emph{stroke} is a pair
$x = (V_0, \varphi) \in \OO = (0, V_{\max}] \times \RR / 2\pi\mathbb{Z}$
with $V_{\max} = 12$ m/s; directions are angles modulo $2\pi$, and we
write them as real numbers. The cue ball is set in motion with velocity
$V_0 (\cos \varphi, \sin \varphi)$, all object balls being at rest.
Whenever the evolution from $(c, x)$ has finitely many events before
all balls are at rest or pocketed, $N(c, x) \in \{0, \dots, 21\}$
denotes the number of object balls pocketed.
\end{definition}

The lengths $L$, $W$ and $R$ are those of a regulation snooker table
and ball. The remaining constants are chosen values, not fitted ones,
and the theorems are statements about Model M with exactly these
values. They are of the size of measured ones: high-speed imaging on a
snooker table gives a rolling deceleration of $0.124$ to $0.126$
m/s$^{2}$ and a cushion restitution, for impacts perpendicular to the
cushion, of about $0.91$ at low incident speeds, with a best linear
fit of $0.82$ over incident speeds from $0.28$ to $3.5$ m/s
\cite{mathavan2009}. We cite no measurement for
$e = 0.95$, and $\rho$ is a pocket radius only in the sense of
Definition~\ref{def:table}. Problem (P5) of Section~\ref{sec:open} asks
which of these choices matter.

The bound $V_{\max}$ sits just above the fastest measured cue-ball
speeds in professional play, about $11.6$ m/s \cite{onoda1989}. We use
$N(c, x)$ only for strokes where finiteness of the evolution is part of
what is proved.

\begin{remark}[The touching rack]
\label{rem:touching}
In the regulation opening position the fifteen reds touch, so the
first impact produces simultaneous multi-ball contacts. Their outcome
is not determined by pairwise collision laws \cite{ballard2000}, and
for the literal rack $N(c_0, x)$ is undefined. We therefore work with
the $\varepsilon$-separated opening configuration $c_0^{\varepsilon}$,
with $\varepsilon = 0.5$ mm: the cue ball at $(0.65, W/2)$, inside the
D; the six colors on their spots, at $(0.737, W/2 \mp 0.292)$,
$(0.737, W/2)$, $(L/2, W/2)$, $(3L/4, W/2)$ and $(L - 0.324, W/2)$; and
the fifteen reds on a triangular lattice of pitch $2R + \varepsilon$,
in rows $r = 0, \dots, 4$ at
$x = 3L/4 + (2R + \varepsilon)(1 + r\sqrt{3}/2)$, row $r$ holding
$r + 1$ balls at $y = W/2 + (j - r/2)(2R + \varepsilon)$,
$0 \le j \le r$. All pairwise gaps are then positive.
\end{remark}

\begin{remark}[Scope]
\label{rem:scope}
Model M omits spin, sliding, throw, air resistance and the angular
acceptance of real pocket jaws. The results below are statements about
Model M. A witness in the spinless model, struck with zero spin,
remains a witness in any enlarged model that agrees with Model M on
the zero-spin slice. The speed-independence behind
Theorem~\ref{thm:ribbon} is specific to the model.
\end{remark}

\section{Clocks and local constancy}
\label{sec:clocks}

A \emph{state} lists, for each active ball, its center and velocity.
For a moving ball $i$ write $p_i(t) = p_i + v_i t + q_i t^{2}$ with
$q_i = -\tfrac{a}{2} v_i / |v_i|$, valid for $0 \le t \le |v_i| / a$,
and $p_i(t) = p_i$ for a ball at rest. Every possible next event is the
first zero of one of the following \emph{clocks}, each a polynomial in
$t$ whose coefficients are polynomial in the state and in the
quantities $1/|v_i|$.
\begin{itemize}[leftmargin=2em]
\item Pair clocks $g_{ij}(t) = |p_i(t) - p_j(t)|^{2} - 4R^{2}$, for
two active balls of which at least one moves. These are of degree at
most four.
\item Cushion clocks $x_i(t) - R$, $L - R - x_i(t)$, $y_i(t) - R$ and
$W - R - y_i(t)$, of degree two, for a moving ball $i$.
\item Pocket clocks $h_{ik}(t) = |p_i(t) - P_k|^{2} - \rho^{2}$, of
degree at most four, for a moving ball $i$.
\item Stop clocks $s_i(t) = |v_i| - a t$, for a moving ball $i$.
\end{itemize}
An event is \emph{transversal} if its clock has nonzero derivative at
its zero. For collisions this means the balls approach with positive
normal speed, and for pocketings that the center crosses the pocket
circle with positive radial speed. Event-based pool simulators find the
next event in the same way, as the smallest root of polynomials of
degree at most four \cite{leckie2006}.

\begin{definition}[Regular stroke]
A stroke $x$ is \emph{regular} for $c$ if the evolution from $(c, x)$
has finitely many events, all transversal and no two simultaneous, and
no ball leaves a collision with velocity zero.
Since a stop is an event, no ball comes to rest exactly at a contact.
\end{definition}

\begin{lemma}[Local constancy]
\label{lem:continuity}
Let $c$ be admissible. The set of strokes regular for $c$ is open in
$\OO$, and on it the list of events, with its participants and its
order, is locally constant. In particular, if $x^{*}$ is regular for
$c$, then $N(c, \cdot)$ is defined and constant on a neighborhood of
$x^{*}$.
\end{lemma}

\begin{proof}
Let $x^{*}$ be regular, with events $E_1, \dots, E_n$ at times
$t_1 < \dots < t_n$, and put $t_0 = 0$. We show by induction on $k$
that for $x$ near $x^{*}$ the first $k$ events are $E_1, \dots, E_k$,
with the same participants, transversal and not simultaneous, and that
their times and the states after them are continuous in $x$. Suppose
this holds for $k - 1$. Between events the state evolves by the
explicit formulas above, which are continuous as long as no ball
stops, and a stop is an event; by regularity a ball that leaves a
collision has nonzero velocity, so near $x^{*}$ the same balls move
after $E_{k-1}$. At $x^{*}$ the clock of $E_k$ is nonzero
on $(t_{k-1}, t_k)$ and has a simple zero at $t_k$. Every other clock
is nonzero on $(t_{k-1}, t_k]$; those of the balls that met at
$E_{k-1}$ vanish at $t_{k-1}$ with nonzero derivative, so their
quotients by $t - t_{k-1}$ are nonzero on $[t_{k-1}, t_k]$. These
finitely many strict inequalities on compact intervals persist under
small changes of the coefficients of the clocks, and the simple zero
moves continuously by the implicit function theorem. Hence for $x$
near $x^{*}$ the next event is $E_k$, transversal and alone, at a time
close to $t_k$, and the continuous event map gives a state continuous
in $x$.

After $E_n$ every ball is at rest or pocketed at $x^{*}$: a ball that
moved has had its stop or its pocketing among the events, and a ball
that never moved is at rest. For $x$ near $x^{*}$ the list of events
and its participants are the same, so the same holds, the evolution
from $x$ ends with $E_n$, and $x$ is regular. The set of pocketed balls
is read off the list.
\end{proof}

Two general facts about the evolution close this section: it cannot
last long, and the strokes with a given list of events form a
semialgebraic set.

\begin{lemma}[Duration]
\label{lem:duration}
Let $c$ be admissible, $x = (V_0, \varphi) \in \OO$, and let $E(t)$ be
half the sum of the squared speeds of the active balls along the
evolution from $(c, x)$. Then $\sqrt{2E(t)} \le V_0 - at$ as long as
some ball moves. In particular no event occurs after time
$V_0 / a \le 120$ s, and if the events are finitely many, every ball
is at rest or pocketed by that time.
\end{lemma}

\begin{proof}
Between events each moving ball has $d|v_i|/dt = -a$, so
$dE/dt = -a \sum |v_i|$, the sum over the moving balls, and
$\sum |v_i| \ge (\sum |v_i|^{2})^{1/2} = \sqrt{2E}$. Hence
$\sqrt{2E}$ decreases at rate at least $a$ while $E > 0$. At a
collision the tangential components are unchanged and the normal
components follow Definition~\ref{def:dynamics}, which changes $E$ by
$-(1 - e^{2})(u_1 - u_2)^{2}/4 \le 0$. A cushion contact multiplies a
velocity component by $e_c < 1$, a stop replaces a vanishing velocity
by zero, and a pocketing removes a ball. So $E$ never increases at an
event, and $\sqrt{2E(0)} = V_0$ gives the inequality. Every event
happens at a time at which some ball moves or has just stopped, so at
a time $t$ with $0 \le V_0 - at$. If the events are finitely many, no
ball moves after the last one, since a moving ball produces a further
event, its stop, before time $V_0 / a$.
\end{proof}

The evolution from $(c, x)$ therefore fails to determine $N(c, x)$ in
two cases only: infinitely many events accumulate at some time not
exceeding $V_0 / a$, or a simultaneous contact of a kind not covered by
Definition~\ref{def:dynamics} occurs, as in
Remark~\ref{rem:touching}. Whether the first case can occur in Model M
is open (Section~\ref{sec:open}).

We regard $\OO$ as the set $(0, V_{\max}] \times S^{1} \subset \RR^{3}$
through $(V_0, \varphi) \mapsto (V_0, \cos\varphi, \sin\varphi)$, and
call a subset of $\OO$ \emph{semialgebraic} if its image is a
semialgebraic subset of $\RR^{3}$, that is, a finite union of sets
defined by polynomial equations and inequalities
\cite{bcr1998}. For an admissible $c$ and a finite list $\Sigma$ of
events, each with its participants, let $F_\Sigma(c)$ be the set of
$x \in \OO$ such that the evolution from $(c, x)$ has exactly the
events of $\Sigma$, each ball undergoing its events in the order in
which they appear in $\Sigma$. Simultaneous events with disjoint
participants are listed in a fixed conventional order, and a stroke
at which the evolution is undefined belongs to no $F_\Sigma(c)$.

\begin{proposition}[Semialgebraic structure]
\label{prop:semialg}
For every admissible $c$ and every finite list $\Sigma$, the set
$F_\Sigma(c)$ is semialgebraic. Consequently it has finitely many
connected components, each semialgebraic, and
$\{\, x \in \OO : N(c, x) = 21 \,\}$ is a countable union of
semialgebraic sets. If the number of events of the evolution from
$(c, x)$ is bounded over $x \in \OO$, this union is finite, and the
set is semialgebraic with finitely many connected components. If the
coordinates of $c$ are rational, as those of $c^{*}$ are, all these
sets are defined over $\mathbb{Q}$.
\end{proposition}

\begin{proof}
We work in the first-order language of ordered fields, with the model
constants and the coordinates of $c$ as parameters. The initial state
is polynomial in $(V_0, \cos\varphi, \sin\varphi)$. Call an
\emph{event time} an instant at which at least one event occurs.
Suppose that the state after the $k$-th event time is a semialgebraic
function of $x$, that is, one with a semialgebraic graph, on the set of
strokes whose events at their first $k$ event times are given. Whether
a ball moves in that state is the polynomial condition $v_i \ne 0$. For
the moving balls the clocks are polynomials in $t$ whose coefficients
are polynomial in the state and in the numbers $1/|v_i|$, and $|v_i|$
is the positive square root of a polynomial, so these coefficients are
semialgebraic functions of $x$. By Definitions~\ref{def:table}
and~\ref{def:dynamics}, that a given event occurs at time $t$ is a
first-order condition on the state and $t$, through the signs of its
clock at $t$ and at nearby times. A zero of a clock need not be an
event: a center may touch a pocket circle without entering the disk.
That the next event time is $\tau$ and that the events at $\tau$ are
those of a given set $\mathcal{E}$, no two with a ball in common, says
that $\tau > 0$, that no event occurs on $(0, \tau)$, and that the
events occurring at $\tau$ are exactly those of $\mathcal{E}$. This is
a first-order formula in $x$, the state and $\tau$, so by the
Tarski--Seidenberg theorem \cite{bcr1998} the set where it holds is
semialgebraic, and $\tau$ is a semialgebraic function on it. The event
maps of Definition~\ref{def:dynamics} are polynomial, the collision law
having normal $(p_j - p_i)/2R$, and those of events with no ball in
common act on different balls, so the state after the $(k+1)$-st event
time is again a semialgebraic function of $x$. That the evolution ends
there, every active ball being at rest, is also a first-order
condition. Hence the set of strokes whose evolution has given events at
each of its event times is semialgebraic, and $F_\Sigma(c)$ is the
union of these sets over the finitely many ways of distributing the
events of $\Sigma$ over successive event times that respect the order
of each ball's events.

A semialgebraic set has finitely many connected components, each of
them semialgebraic \cite{bcr1998}. The set $\{N(c, \cdot) = 21\}$ is
the union of the sets $F_\Sigma(c)$ over the countably many lists in
which twenty-one object balls are pocketed, and over finitely many if
the number of events is bounded. When the coordinates of $c$ are
rational, every parameter in the formulas is rational, since the model
constants are.
\end{proof}

\section{Certification}
\label{sec:cert}

The following lemma turns one event of the dynamics into a finite list
of inequalities that can be checked on a whole set of states at once.
Let $f$ be a clock with $f(0) = 0$ for every state under consideration,
as happens for the pair that has just collided or the ball that has
just met a cushion. We call such a clock \emph{fresh}, and write
$f(t) = t\, \tilde f(t)$, where the polynomial $\tilde f$ has
coefficients that do not involve the constant term of $f$.

\begin{lemma}[One certified step]
\label{lem:step}
Let $\mathcal{K}$ be a set of states with the same active and moving
balls, let $E$ be a clock, and let $0 < \alpha < \beta$. Suppose that
for every state in $\mathcal{K}$ the following hold.
\begin{enumerate}[label=(\alph*), leftmargin=2.5em]
\item Every stop clock other than $E$ is positive at $\beta$.
\item $E > 0$ on $[0, \alpha]$, $E(\beta) < 0$, and $E' < 0$ on
$[\alpha, \beta]$; or, if $E$ is fresh, the same three conditions for
$\tilde E$.
\item Every other clock $f$ is positive on $[0, \beta]$, or, if $f$
is fresh, $\tilde f$ is positive on $[0, \beta]$.
\end{enumerate}
Then from every state in $\mathcal{K}$ the next event is the zero of
$E$, at a unique time $\tau \in (\alpha, \beta)$. It is transversal,
and no other event occurs at time $\tau$.
\end{lemma}

\begin{proof}
By (b), $E$ is positive on $(0, \alpha]$, strictly decreasing on
$[\alpha, \beta]$ and negative at $\beta$, so it has exactly one zero
$\tau$ in $(0, \beta]$, which lies in $(\alpha, \beta)$ and satisfies
$E'(\tau) < 0$. If $E$ is fresh, the same holds for $\tilde E$; since
$E = t \tilde E$, the zero $\tau$ of $\tilde E$ is the only zero of $E$
in $(0, \beta]$, and $E'(\tau) = \tau \tilde E'(\tau) < 0$. If $E$ is a
stop clock, $\tau$ is the stopping time of
its ball. By (a) no other ball stops before $\beta$, so on $[0, \tau]$
every center follows its polynomial formula and every clock is the
stated polynomial. By (c) no other clock vanishes on $(0, \tau]$. The
next event is therefore $E$, at $\tau$, it is transversal, and it is
not simultaneous with any other.
\end{proof}

For a set $\mathcal{K}$ of states and an interval $[\alpha, \beta]$,
let $\Phi_E(\mathcal{K}, [\alpha, \beta])$ denote the set of states
obtained by flowing a state of $\mathcal{K}$ for a time in
$[\alpha, \beta]$ and applying the event map of $E$.

\begin{corollary}[Certified chain]
\label{cor:chain}
Let $X \subset \OO$, let $\mathcal{K}_0$ contain the initial states
$(c, x)$ for $x \in X$, and let $E_1, \dots, E_n$ be events with
intervals $[\alpha_k, \beta_k]$ and sets $\mathcal{K}_k \supseteq
\Phi_{E_k}(\mathcal{K}_{k-1}, [\alpha_k, \beta_k])$. Suppose that
for each $k$ the hypotheses of Lemma~\ref{lem:step} hold for
$\mathcal{K}_{k-1}$, $E_k$ and $[\alpha_k, \beta_k]$, where at step
$k \ge 2$ the clocks of the event $E_{k-1}$ that vanish at the
post-event state are treated as fresh. Then for every $x \in X$ the
first $n$ events from $(c, x)$ are $E_1, \dots, E_n$, in this order,
each transversal and none simultaneous with another.
\end{corollary}

\begin{proof}
Induction on $k$. The state of the true evolution just after its
$(k-1)$-st event lies in $\mathcal{K}_{k-1}$. The fresh clocks vanish
there exactly, because the two balls of a collision are at distance
$2R$ and a ball that met a cushion lies on the cushion line, so
$\tilde f$ is the correct quotient. Lemma~\ref{lem:step} identifies
the next event as $E_k$ at a time in $[\alpha_k, \beta_k]$, and the
new state lies in $\Phi_{E_k}(\mathcal{K}_{k-1}, [\alpha_k, \beta_k])
\subseteq \mathcal{K}_k$.
\end{proof}

\paragraph{Verification in interval arithmetic.}
The hypotheses of Corollary~\ref{cor:chain} are verified in interval
arithmetic with outward rounding \cite{moore2009,tucker2011}. The model
constants are enclosed as exact decimals, and the configuration as
exact binary64 numbers. Each $\mathcal{K}_k$ is a box, a product of
intervals for the positions and velocities. Clock coefficients are
enclosed by evaluating their formulas on the box. Positivity of a
polynomial on $[0, \beta]$ is proved by bisection of the time interval,
using the natural interval extension and the mean value form on each
piece. The interval $[\alpha_k, \beta_k]$ is found by testing the sign
conditions of (b) at points around an approximate root. For a
collision, the event map uses the exact unit normal
$(p_j - p_i)/(2R)$ at contact. For a cushion contact, the coordinate
normal to the cushion is set to its exact value. Which event to test
next is proposed by a floating-point computation at the center of the
box. The proposal affects only whether a run succeeds, never what a
successful run proves.

Two programs certify the box, built on different
interval libraries, mpmath \cite{mpmath} and Arb \cite{johansson2017}.
The first implements the procedure above. The second was written from
the model's specification alone, shares no code with the first, and
follows a different reduction: it encloses the first root of a clock in
closed form where the clock is quadratic in the arc length of the
moving ball, isolates it by bisection and an interval Newton step
otherwise, and shows that the enclosure of the earliest event lies
strictly before those of all the others.

\section{A certified box}
\label{sec:box}

\begin{theorem}
\label{thm:box}
Let $X = [V_0^{*} - 10^{-90}, V_0^{*} + 10^{-90}] \times
[\varphi^{*} - 10^{-90}, \varphi^{*} + 10^{-90}]$. For every
$x \in X$, the evolution from $(c^{*}, x)$ up to the pocketing of the
last object ball consists of 61 events, namely 21 ball-ball
collisions, 19 cushion contacts and 21 pocketings, in an order that does
not depend on $x$. All are transversal and no two are simultaneous.
Moreover $N(c^{*}, x) = 21$. In particular
$\{\, x \in \OO : N(c^{*}, x) = 21 \,\}$ has positive Lebesgue measure.
\end{theorem}

\begin{proof}
The first program verifies the hypotheses of
Corollary~\ref{cor:chain} for
$X$ and the 61 events, at 1200 bits of precision. The second certifies,
by its own reduction, the same 61 events in the same order. After the
61st event only the cue ball is on
the table. A single decelerating ball travels a finite distance, and
among any three consecutive cushion contacts it crosses the table, so
it meets the cushions finitely often before it stops or is pocketed.
The evolution therefore has finitely many events, and
$N(c^{*}, x) = 21$.
\end{proof}

The computation also yields explicit margins, valid for every stroke
in $X$. Every ball-ball collision has normal approach speed at least
$0.50$ m/s, every cushion contact at least $0.10$ m/s, and every ball
enters its pocket with speed at least $0.44$ m/s. The box is proved
directly, so Lemma~\ref{lem:continuity} is not needed for
Theorem~\ref{thm:box}. It explains why such a box must exist around any
regular stroke. Since $X$ is so thin, the stroke and configuration must
be read as exact binary64 numbers (Appendix~\ref{app:witness}).

\section{A certified ribbon}
\label{sec:ribbon}

In the witness every collision is a moving ball striking a ball at
rest. Such collisions, and cushion contacts, transform directions
independently of speeds.

\begin{lemma}[Speed independence]
\label{lem:speed}
Let a ball with velocity $s u$, $|u| = 1$, $s > 0$, strike a ball at
rest, with unit normal $n$ from the first center to the second and
$u \cdot n > 0$. Then afterwards the struck ball has velocity
$k s n$ and the striking ball has velocity $s (u - k n)$, where
$k = \tfrac{1+e}{2} (u \cdot n)$. A cushion contact maps a velocity
$s u$ to $s u'$, where $u'$ depends only on $u$. Between events a
moving ball travels along the ray of its velocity.
\end{lemma}

\begin{proof}
The first statement is the collision law of
Definition~\ref{def:dynamics} with $u_2 = 0$. The others are immediate
from the same definition.
\end{proof}

Consequently, as long as the same events occur, every ball traces the
same polyline whatever the speed of the stroke, and only the times at
which it passes along it change. A \emph{skeleton} for $c^{*}$ and a
direction is the list of these polylines. For each ball it
consists of a \emph{rest piece}, the ball's initial position, up to
the event at which it is struck, followed by straight \emph{moving
pieces}, each with a starting point, a unit direction, a length, and
a designated end event (a collision with a ball at rest, a cushion
contact, or a pocketing). The skeleton ends, for the cue ball, with its
last designated event. After that event the cue ball is in \emph{free
motion}: it moves by the deceleration law and the cushion law of
Definition~\ref{def:dynamics} until it stops or is pocketed, unless it
meets another ball. Its free motion depends on the speed of the stroke
and is not part of the skeleton. For a stroke $x$, the speeds along the
pieces follow from Lemma~\ref{lem:speed} and the deceleration law, and
so do the time windows during which a ball is on each piece. A set of
strokes has a \emph{common skeleton} if the skeletons of their
directions have the same pieces with the same designated end events;
the starting points, directions and lengths of the pieces vary with the
direction.

Speeds along a skeleton obey a simple law.

\begin{lemma}[Speed scaling]
\label{lem:scaling}
Fix a skeleton. To each moving piece $\pi$, of length $\ell_\pi$, there
correspond numbers $\alpha_\pi > 0$ and $\beta_\pi \ge 0$, determined
by the skeleton alone, such that for every stroke $x = (V_0, \varphi)$
whose evolution follows the skeleton, the speed $v$ of the ball at arc
length $s$ along $\pi$ satisfies
\[
v^{2} = \alpha_\pi V_0^{2} - \beta_\pi - 2 a s,
\qquad 0 \le s \le \ell_\pi .
\]
The first piece of the cue ball has $\alpha = 1$ and $\beta = 0$. If
$\pi$, with unit direction $u$, ends in a collision with a ball at rest
with unit normal $n$, and $k = \tfrac{1+e}{2}(u \cdot n)$, then the
struck ball's piece has
$(\alpha, \beta) = k^{2}(\alpha_\pi, \beta_\pi + 2a\ell_\pi)$ and the
striking ball's next piece has
$(\alpha, \beta) = |u - kn|^{2}(\alpha_\pi, \beta_\pi + 2a\ell_\pi)$.
If $\pi$ ends at a cushion that maps $u$ to $u'$, the next piece has
$(\alpha, \beta) = |u'|^{2}(\alpha_\pi, \beta_\pi + 2a\ell_\pi)$. In
particular every speed along the skeleton is a strictly increasing
function of $V_0$, and a stroke can follow the skeleton only if
$V_0 > V_c$, where
\[
V_c^{2} = \max_{\pi} \frac{\beta_\pi + 2a\ell_\pi}{\alpha_\pi},
\]
the maximum being over the moving pieces of the skeleton.
\end{lemma}

\begin{proof}
Along a piece the deceleration law gives $v(t) = v(0) - at$ and
$s(t) = v(0)t - \tfrac{a}{2}t^{2}$, hence $v^{2} = v(0)^{2} - 2as$. At
the end of the piece, where the speed is $v_1$, Lemma~\ref{lem:speed}
gives the struck ball the speed $k v_1$, the striking ball the speed
$|u - kn|\, v_1$, and after a cushion the speed $|u'|\, v_1$. The three
factors depend only on the skeleton and are positive: $k > 0$ because
$u \cdot n > 0$, $|u - kn| \ge 1 - \tfrac{1+e}{2} > 0$, and
$|u'| \ge e_c$. Induction over the events of the skeleton gives the
formula, with $\alpha_\pi$ the product of the squared factors met so
far and $\beta_\pi$ the accumulated $2a\ell$ terms, each multiplied by
the squared factors met after it. Since $\alpha_\pi > 0$, every speed
increases strictly with $V_0$. A stroke that follows the skeleton has
every ball still moving when it reaches the designated end of each of
its pieces, that is $\alpha_\pi V_0^{2} - \beta_\pi - 2a\ell_\pi > 0$
for every $\pi$, which is $V_0 > V_c$.
\end{proof}

\begin{proposition}[Skeleton criterion]
\label{prop:skeleton}
Let $S$ be a set of strokes with a common skeleton, and suppose that
for every $x \in S$ the following hold.
\begin{enumerate}[label=(S\arabic*), leftmargin=2.8em]
\item Along each moving piece the ball's speed stays positive, and no
cushion line or pocket disk is reached before the designated end
event, which occurs transversally at the end of the piece.
\item A piece that ends in a collision with a ball at rest reaches
distance $2R$ from that ball for the first time at its end, and
transversally.
\item Two pieces of different balls that do not meet at a common event
keep their balls more than $2R$ apart at every instant at which both
balls are on them.
\item The two moving pieces that begin at a common collision keep their
balls more than $2R$ apart at every positive time at which both balls
are on them.
\item After its last designated event the cue ball, in its free motion,
is at every later instant more than $2R$ from every object ball that is
then on the table.
\end{enumerate}
Then for every $x \in S$ the evolution from $(c^{*}, x)$ follows the
skeleton. Every ball moves along its pieces, the cue ball then
continuing in free motion, the events are exactly the designated ones
and those of the free motion of the cue ball, and every object ball is
pocketed.
\end{proposition}

\begin{proof}
Fix $x \in S$ and order the designated events by their times for $x$.
Suppose the evolution agrees with the skeleton motion up to and
including some designated event, the cue ball being in free motion if
its last designated event has occurred. Until the next designated
event, each ball moves along its current piece by the deceleration law,
and the cue ball, once in free motion, continues it. An event of the
free motion involves the cue ball alone and leaves the other balls on
their pieces. Any other event in this interval would be a stop, an
undesignated cushion or pocket contact of a ball on a piece, or an
undesignated ball-ball contact. The first two are excluded by (S1). An
undesignated contact of two balls on pieces is excluded by (S3), by
(S4) for the pair that has just collided, and by (S2) for a moving ball
approaching the ball it is about to strike. A contact with the cue ball
in free motion is excluded by (S5). So the next event, apart from those
of the free motion, is designated. It is transversal by (S1) or (S2),
and by Lemma~\ref{lem:speed} it sends the balls involved onto their
next pieces, or the cue ball into free motion. Induction over the
finitely many designated events completes the proof.
\end{proof}

\begin{theorem}
\label{thm:ribbon}
Every stroke $x = (V_0, \varphi) \in \OO$ with
$V_0 \ge V_0^{*} - 2.3$ and $|\varphi - \varphi^{*}| \le 10^{-90}$
satisfies $N(c^{*}, x) = 21$.
\end{theorem}

\begin{proof}
A third program builds the skeleton in interval arithmetic over the
box of directions, and encloses, over a range of speeds, the time
windows of all pieces. It verifies (S1) and (S2) directly, and (S3)
by subdividing the common time windows and bounding the distance
between the two balls on each sub-window. It verifies (S4) as for
fresh clocks, dividing the squared distance minus $4R^2$ by $t$. It
verifies (S5) by following the free motion of the cue ball, through its
cushion contacts, until it comes to rest or every object ball is
pocketed, and by treating the segments of that motion as further pieces
in (S3) and (S4). The
speed range $[V_0^{*} - 2.3, V_0^{*} + 0.1]$, which covers
$[V_0^{*} - 2.3, V_{\max}]$, was split into 48 sub-intervals of width
$0.05$ m/s, each verified separately. Proposition~\ref{prop:skeleton}
then shows that every object ball is pocketed, and the motion of the
cue ball has finitely many events as in the proof of
Theorem~\ref{thm:box}.
\end{proof}

\section{The width of the ribbon}
\label{sec:thin}

Theorems~\ref{thm:box} and~\ref{thm:ribbon} certify a ribbon of angular
half-width $10^{-90}$. The true window of directions is far wider, and
still astronomically thin. In this section we determine it, with
certificates of the same kind, and then explain the number.

Let $\Sigma$ denote the \emph{event list} of the witness: the 61 events
of Theorem~\ref{thm:box}, listed in Table~\ref{tab:events}, each with
its participants (the cue ball and
the object ball it strikes, a ball and the cushion it meets, a ball and
the pocket it enters). A stroke $x$ \emph{realizes} $\Sigma$ if, in the
evolution from $(c^{*}, x)$, every object ball undergoes exactly its
events of $\Sigma$, in the order in which they appear in $\Sigma$, and
the cue ball undergoes first its events of $\Sigma$, in that order.
Nothing is prescribed for the cue ball after its last event of
$\Sigma$, the collision with ball 21: it rolls on, may meet cushions,
and comes to rest or is pocketed, and this may happen before the last
object ball is pocketed. Let
\[
F = \{\, x \in \OO : x \text{ realizes } \Sigma \,\},
\]
and let $F^{*}$ be the connected component of $F$ containing $x^{*}$.
Every $x \in F$ has $N(c^{*}, x) = 21$, by the argument that ends the
proof of Theorem~\ref{thm:box}, and the proof of
Theorem~\ref{thm:ribbon} shows that $F$
contains the ribbon. Since the events of
$\Sigma$ occur in different orders for different speeds, $F$ is defined
through the order of each ball's own events, not the global order.

In $\Sigma$ every collision is the cue ball striking an object ball that
has had no previous event, hence a ball at rest. Lemma~\ref{lem:speed}
therefore determines the geometry of a stroke realizing $\Sigma$ from
its direction alone, by the following construction. Given $\varphi$,
the cue ball starts at its position with direction
$(\cos\varphi, \sin\varphi)$, and the events of $\Sigma$ are processed
in order. At a collision, the moving ball's ray is followed to the
first point at distance $2R$ from the designated resting ball, if
there is one; the contact point gives the unit normal, and
Lemma~\ref{lem:speed} gives the two outgoing directions. At a cushion
contact the ray is followed to the designated cushion line and the
direction is reflected there. At a pocketing the ray is followed to the
designated pocket disk. The construction \emph{succeeds} for $\varphi$
if every designated contact is reached, and the resulting polylines are
the \emph{skeleton} of $\varphi$.

\begin{lemma}[Skeleton of a direction]
\label{lem:skeleton}
Let $x = (V_0, \varphi)$ realize $\Sigma$. Then the construction
succeeds for $\varphi$, every ball moves along its pieces of the
skeleton of $\varphi$, and the designated end of each moving piece is
the first point of its ray, beyond its starting point, at which the ray
meets a cushion line, enters a pocket disk, or reaches distance $2R$
from the designated ball.
\end{lemma}

\begin{proof}
Induction along the evolution. The cue ball leaves along
$(\cos\varphi, \sin\varphi)$. Suppose the evolution agrees with the
construction up to some event of $\Sigma$. The ball or balls set in
motion by that
event leave along the directions of Lemma~\ref{lem:speed}, which depend
only on the incoming direction and the contact geometry, both known from
the construction. Each ball that has a further event in $\Sigma$ then
moves along its ray until its next
event, which by assumption is its next event in $\Sigma$. Had the ray
met a cushion line, entered a pocket disk, or reached distance $2R$
from the designated ball earlier, a cushion contact, a pocketing, or
the collision would have occurred earlier, and the evolution would not
realize $\Sigma$. So the designated contact is reached, and it is the
first of the three kinds of points along the ray.
\end{proof}

The construction is a finite sequence of algebraic operations on
$\varphi$, and the condition that the designated end comes first is a
finite list of inequalities. Both can be evaluated in interval
arithmetic over a set of directions. If, for every $\varphi$ in the
set, some inequality provably fails, then by Lemma~\ref{lem:skeleton}
no stroke with direction in the set realizes $\Sigma$, whatever its
speed. If instead every inequality provably holds, the skeleton is
realized geometrically for every direction in the set, and
Proposition~\ref{prop:skeleton} reduces membership in $F$ to its timing
conditions.

\paragraph{Centered forms.}
The box method of Sections~\ref{sec:cert} and~\ref{sec:ribbon} cannot
cover the window. The dynamics amplifies a change of direction by a
factor of about $3 \times 10^{40}$ in the position of the last ball,
and a box of directions evaluated in interval arithmetic loses many
orders of magnitude more through the dependency problem: boxes of
half-width $10^{-88}$ are already undecided, and the window, of width
about $10^{-42}$, would need some $5 \times 10^{47}$ boxes of the
certified half-width $10^{-90}$. The remedy is a first-order
enclosure in the single variable $\varphi$. Fix a piece
$\Phi = [\varphi_c - h, \varphi_c + h]$ and write $t = \varphi -
\varphi_c$. A \emph{centered form} for a function $q$ on $\Phi$ is a
pair $(c, D)$ of intervals with
\[
q(\varphi) \in c + D\, t \qquad \text{for all } \varphi \in \Phi .
\]

\begin{lemma}[Centered forms]
\label{lem:centered}
Let $(c_1, D_1)$ and $(c_2, D_2)$ be centered forms for $q_1$ and
$q_2$ on $\Phi$, let $T = [-h, h]$, and let $\rho(c, D) = c + D T$
denote the range of a form. Then centered forms for $q_1 \pm q_2$,
$q_1 q_2$, $\sqrt{q_1}$ and $1/q_1$ are given by
\begin{align*}
&(c_1 \pm c_2,\; D_1 \pm D_2), \qquad
(c_1 c_2,\; c_1 D_2 + c_2 D_1 + D_1 D_2 T), \\
&\Bigl(\sqrt{c_1},\; \frac{D_1}{\sqrt{\rho(c_1, D_1)} + \sqrt{c_1}}\Bigr)
\ \text{if } \rho(c_1, D_1) > 0, \qquad
\Bigl(\frac{1}{c_1},\; \frac{-D_1}{c_1\, \rho(c_1, D_1)}\Bigr)
\ \text{if } 0 \notin \rho(c_1, D_1),
\end{align*}
and $(\cos\varphi_c, -\sin\Phi)$, $(\sin\varphi_c, \cos\Phi)$ are
centered forms for $\cos\varphi$ and $\sin\varphi$.
\end{lemma}

\begin{proof}
Write $q_i(\varphi) = c_i' + D_i' t$ with $c_i' \in c_i$ and
$D_i' \in D_i$. The sum and the product are the identities
$q_1 q_2 = c_1' c_2' + (c_1' D_2' + c_2' D_1' + D_1' D_2' t)\, t$. For
the square root, $\sqrt{q_1} - \sqrt{c_1'} = (q_1 - c_1') /
(\sqrt{q_1} + \sqrt{c_1'})$ with $q_1 - c_1' = D_1' t$ and $q_1 \in
\rho(c_1, D_1)$. For the reciprocal, $1/q_1 - 1/c_1' = -(q_1 - c_1') /
(c_1' q_1)$. For the cosine and sine, the mean value theorem on
$\Phi$.
\end{proof}

In a centered form the width of the range is $2h\,|D|$ up to the width
of $D$ itself, which is of second order in $h$. The enclosure of the
construction is therefore as sharp as the dynamics itself, for pieces
far wider than any box that the box method decides. The box width $h$
enters only through the
product term $D_1 D_2 T$ and through the ranges in the square root and
the reciprocal.

\begin{theorem}[The window of directions]
\label{thm:width}
Let
\[
\delta_{+} = 9.0454 \times 10^{-43}, \quad
\delta_{-} = 4.3586 \times 10^{-44}, \quad
\varepsilon_{+} = 9.0455 \times 10^{-43}, \quad
\varepsilon_{-} = 4.3587 \times 10^{-44}
\]
radians, the exact binary values being those recorded in the
certificates. Then the following hold.
\begin{enumerate}[label=(\roman*), leftmargin=2.5em]
\item $[V_0^{*} - 2.3, V_{\max}] \times
[\varphi^{*} - \delta_{-}, \varphi^{*} + \delta_{+}] \subset F$.
\item No stroke whose direction is within $10^{-90}$ of
$\varphi^{*} + \varepsilon_{+}$, or within $10^{-90}$ of
$\varphi^{*} - \varepsilon_{-}$, lies in $F$, whatever its speed.
\item Consequently the projection of $F^{*}$ to the direction
coordinate is an interval $J$ with
\[
[\varphi^{*} - \delta_{-}, \varphi^{*} + \delta_{+}] \subset J \subset
(\varphi^{*} - \varepsilon_{-}, \varphi^{*} + \varepsilon_{+}),
\]
and the length of $J$ lies between $9.48126 \times 10^{-43}$ and
$9.48137 \times 10^{-43}$, so it is $9.481 \times 10^{-43}$ to four
significant figures.
\end{enumerate}
In particular $\{\, x \in \OO : N(c^{*}, x) = 21 \,\}$ has Lebesgue
measure at least $(V_{\max} - V_0^{*} + 2.3)(\delta_{+} + \delta_{-})
> 2.27 \times 10^{-42}$, in m/s times radians.
\end{theorem}

\begin{proof}
(i) The window $[-\delta_{-}, \delta_{+}]$ of offsets is covered by a
finite number of pieces $\Phi$, chosen by bisection, and the speed
range $[V_0^{*} - 2.3, V_0^{*} + 0.1] \supset [V_0^{*} - 2.3,
V_{\max}]$ by 48 intervals of length $0.05$ m/s. For each piece and
each speed interval, the construction of the skeleton is carried out
in centered forms over $\Phi$ by Lemma~\ref{lem:centered}, the
designated events being those of $\Sigma$, and the geometric
conditions, (S2) and the part of (S1) concerning cushions and pockets,
are verified on the ranges. The timing conditions, the positivity of
speeds in (S1), and (S3), (S4) and (S5), are then verified exactly as in the
proof of Theorem~\ref{thm:ribbon}, with the ranges of the centered
forms in place of the enclosures of the box method. By
Proposition~\ref{prop:skeleton}, every stroke in the piece and speed
interval realizes $\Sigma$. Twenty-three pieces suffice for each speed
interval, the widest $2.4 \times 10^{-43}$ wide and the narrowest, near
the two edges, $4.6 \times 10^{-46}$. The same certificates show that
on this rectangle the cue ball, in its free motion, meets no cushion
before the last pocketing, and that it comes to rest after the last
pocketing when $V_0 \ge V_0^{*} - 0.25$ and before it when
$V_0 \le V_0^{*} - 1.2$.

(ii) By Lemma~\ref{lem:skeleton} it suffices to show that for every
direction in the box the designated end of some piece is not the first
of the three kinds of points along its ray. Interval arithmetic over the
box $[\varphi^{*} + \varepsilon_{+} - 10^{-90},
\varphi^{*} + \varepsilon_{+} + 10^{-90}]$ shows that the construction
succeeds through the first 55 events of $\Sigma$, and that at the 56th,
the pocketing of ball 21 in the corner pocket at $(L, 0)$, the ray of
ball 21 reaches the cushion line $x = L - R$ after an arc length of at
most $1.9041038$ m, while it reaches the pocket disk after an arc
length of at least $1.9041041$ m. On the other side, over the box about
$\varphi^{*} - \varepsilon_{-}$, the same ray reaches the cushion line
$y = R$ after at most $1.89578041$ m and the pocket disk after at least
$1.89578042$ m. In both cases the designated pocketing is not the first
point along the ray, so no stroke in the box realizes $\Sigma$.

(iii) The projection of the connected set $F^{*}$ to the direction
coordinate is connected, hence an interval $J$. By (i) the rectangle,
which is connected and contains $x^{*}$, lies in $F^{*}$, so $J$
contains $[\varphi^{*} - \delta_{-}, \varphi^{*} + \delta_{+}]$. By
(ii) $J$ contains neither $\varphi^{*} + \varepsilon_{+}$ nor
$\varphi^{*} - \varepsilon_{-}$. Its length lies between
$\delta_{+} + \delta_{-} = 9.48126 \times 10^{-43}$ and
$\varepsilon_{+} + \varepsilon_{-} = 9.48137 \times 10^{-43}$. The
measure bound follows from (i), since $F \subset \{N = 21\}$.
\end{proof}

\begin{remark}[Event list versus pocketing]
\label{rem:eventlist}
Just outside the window the ray of ball 21 meets a cushion line inside
the jaws of the corner pocket before its center enters the pocket disk.
In Model M the ball rebounds there and may still be pocketed. The
theorem therefore determines the window of strokes that realize
$\Sigma$, not the set $\{N(c^{*}, \cdot) = 21\}$, which near $x^{*}$
may be somewhat wider than $F^{*}$. Every stroke in it that is not in
$F$ follows a different event list.
\end{remark}

\paragraph{Why these numbers.}
The rest of this section explains the numbers of
Theorem~\ref{thm:width} and is not used in any proof: its formulas are
first-order approximations, evaluated in floating-point arithmetic
(Appendix~\ref{sec:repro}).
Both edges are decided by one gate. Ball 21 must enter the disk of
radius $\rho$ about the corner $(L, 0)$ before its center crosses
either cushion line $x = L - R$ or $y = R$. The gate is the arc of the
pocket circle between $(L - R, \sqrt{\rho^{2} - R^{2}})$ and
$(L - \sqrt{\rho^{2} - R^{2}}, R)$, whose width perpendicular to the
ball's path, which runs at about $30^{\circ}$ below the long axis, is
$30.2$ mm. The window of directions is $9.48 \times 10^{-43}$ wide.
Their ratio, about $3.2 \times 10^{40}$ m/rad, is the sensitivity of
the position of ball 21 at its pocket to the direction of the stroke.
The asymmetry of the window, $\delta_{+} / \delta_{-} \approx 21$,
records where in the gate the ball passes at the stroke $\varphi^{*}$.

A first-order computation accounts for this sensitivity, and in fact
reproduces it. Describe a perturbation of the cue ball's line, as in
ray optics, by its lateral offset $\delta$ at a reference point and its
direction error $\varepsilon$; along a free path of length $s$ the pair
becomes $(\delta + s\varepsilon, \varepsilon)$. Consider a collision at
cut angle $\theta$, the angle between the incoming direction and the
line of centers. The signed impact parameter $b = \pm 2R \sin\theta$ of
the line changes by $db = -\delta - 2R\cos\theta\, \varepsilon$, with
$\delta$ and $\varepsilon$ taken at the contact point. By
Lemma~\ref{lem:speed} the outgoing direction of the striking ball makes
the angle $\psi$ with the incoming one, where
\[
\tan \psi = \frac{c \sin\theta \cos\theta}{1 - c \cos^{2}\theta},
\qquad
\frac{d\psi}{d\theta} =
- \frac{c\, (1 - (2 - c) \cos^{2}\theta)}
{\sin^{2}\theta + (1 - c)^{2} \cos^{2}\theta},
\qquad c = \tfrac{1+e}{2},
\]
so that $d\psi / d\theta$ lies between $-c$ and $-0.997\, c$ for
$72^{\circ} \le \theta \le 84^{\circ}$. The outgoing line has, at the
contact point, the offset $-g(\theta)\, db$ and the direction error
$\varepsilon - \psi'(\theta)\, db / (2R\cos\theta)$, where
$g(\theta) = (1 - c) / \sqrt{\sin^{2}\theta + (1 - c)^{2}\cos^{2}\theta}$
is below $0.03$ for these angles. The struck ball leaves along the
normal at the contact point, whose direction turns by
$-\delta / (2R\cos\theta)$. A cushion met at the angle $\alpha$ to its
normal and left at the angle $\alpha'$, with
$\tan\alpha' = \tan\alpha / e_c$, maps $(\delta, \varepsilon)$ to
$(-\delta \cos\alpha' / \cos\alpha,\;
-\varepsilon \cos^{2}\alpha' / (e_c \cos^{2}\alpha))$. These linear
maps, composed along the skeleton of $\varphi^{*}$ from the stroke to
the contact with ball 21 and followed by its path of $1.896$ m to the
pocket, give $3.184 \times 10^{40}$ m/rad. This agrees to four figures
with the exact derivative, with respect to $\varphi$, of the position
of ball 21 at its pocket, which the slopes of the centered forms of
Lemma~\ref{lem:centered} provide.

\begin{figure}[ht]
\centering
\includegraphics[width=0.72\linewidth]{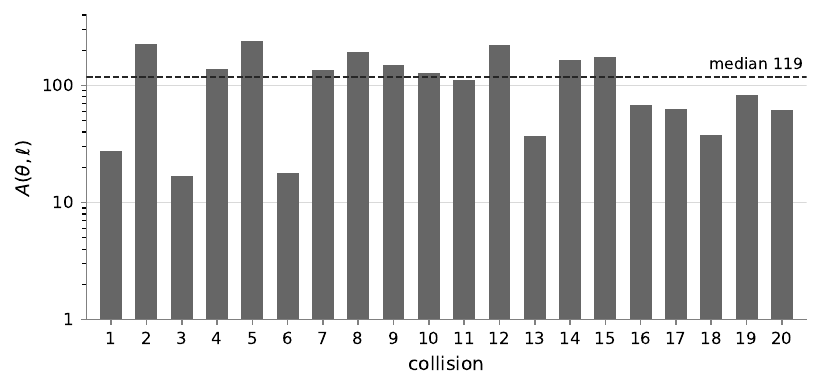}
\caption{The factors $A(\theta, \ell)$ at the first twenty collisions
of the witness $c^{*}$, in order, on a logarithmic scale. The dashed
line marks their median. Their product is $6.7 \times 10^{38}$.}
\label{fig:amp}
\end{figure}

The size of the number comes from one term. At a thin cut $\cos\theta$
is small, an offset $\delta$ becomes a direction error of about
$c\, \delta / (2R\cos\theta)$, and the next free path, of length
$\ell$, turns it into the offset $A(\theta, \ell)\, \delta$ with
\[
A(\theta, \ell) = \frac{c\, \ell}{2R \cos\theta}.
\]
Thin cuts are exactly what keeps the cue fast, and exactly what makes
$A$ large. In the witness the cut angles at the first twenty
collisions lie between $72^{\circ}$ and $84^{\circ}$, the factors $A$
range from 16.8 to 240, with median 119 (Figure~\ref{fig:amp}), and
their product is
$6.7 \times 10^{38}$. The cue travels $1.124$ m to its first contact,
so an aiming error $\delta\varphi$ arrives there as an offset
$1.124\, \delta\varphi$. The twenty-first collision is nearly head-on,
with $\theta = 3.4^{\circ}$, so an offset $\delta$ there turns the
direction of ball 21 by $\delta / (2R\cos\theta)$, which over the
$1.896$ m to the pocket becomes $36.2\, \delta$. The product of these
dominant terms, $2.7 \times 10^{40}$ m/rad, is within fifteen percent
of the exact value; the rest is the small terms above.
The same chain of thin cuts that makes the construction
possible makes it fragile, the mechanism that produces non-uniform
hyperbolicity near tangential collisions in dispersing billiards
\cite{tabachnikov2005}.

The fragility comes from the depth of the chain, twenty thin cuts in
series, not from the number of balls. Section~\ref{sec:cluster}
exhibits a configuration with fuller cuts and shorter legs, in which
no object ball is more than three collisions away from the cue ball,
and its window is thirty-nine orders of magnitude wider.

\section{The lower edge of the ribbon}
\label{sec:edge}

Theorem~\ref{thm:width} describes the window of directions. We now
determine where the ribbon ends in speed. Write
$J_\delta = [\varphi^{*} - \delta_{-}, \varphi^{*} + \delta_{+}]$ for
the window of Theorem~\ref{thm:width}(i). For $\varphi \in J_\delta$
the construction of Section~\ref{sec:thin} succeeds, and we write
$V_c(\varphi)$ for the critical speed of Lemma~\ref{lem:scaling} for
the skeleton of $\varphi$. By Lemma~\ref{lem:skeleton}, a stroke
$(V_0, \varphi)$ realizes $\Sigma$ only if its evolution follows that
skeleton, hence only if $V_0 > V_c(\varphi)$.

\begin{theorem}[The lower edge of the ribbon]
\label{thm:edge}
\sloppy
With $J_\delta$ and $V_c(\varphi)$ as above, the following hold.
\begin{enumerate}[label=(\roman*), leftmargin=2.5em]
\item For every $\varphi \in J_\delta$ the maximum in
Lemma~\ref{lem:scaling} is attained by the last piece of ball 19, from
its rebound off the cushion $x = R$ to the corner pocket at $(0, 0)$,
and $V_c(\varphi)$ lies between $9.32035$ and $9.32132$ m/s. At the
witness direction,
$V_c(\varphi^{*}) = 9.3212701684$ m/s to eleven significant figures.
No stroke $(V_0, \varphi)$ with $\varphi \in J_\delta$ and
$V_0 \le V_c(\varphi)$ lies in $F$.
\item $[9.321315, V_{\max}] \times J_\delta \subset F$.
\item Every stroke with $9.3212701695 \le V_0 \le V_{\max}$ and
$|\varphi - \varphi^{*}| \le 10^{-90}$ lies in $F$.
\end{enumerate}
Consequently, for every $\varphi \in J_\delta$ the set of speeds $V_0$
with $(V_0, \varphi) \in F$ contains $[9.321315, V_{\max}]$ and is
contained in $(V_c(\varphi), V_{\max}]$. The lower edge of the ribbon,
the infimum of these speeds, therefore lies between $V_c(\varphi)$ and
$9.321315$ m/s: it is located to within $10^{-3}$ m/s across the
window, and at $\varphi^{*}$ it lies within $1.1 \times 10^{-9}$ m/s
above $V_c(\varphi^{*})$. The set
$\{\, x \in \OO : N(c^{*}, x) = 21 \,\}$ has
Lebesgue measure at least
\[
(V_{\max} - 9.321315)(\delta_{+} + \delta_{-}) > 2.53 \times 10^{-42},
\]
in m/s times radians.
\end{theorem}

\begin{proof}
(i) On each of the twenty-three pieces $\Phi$ that cover $J_\delta$ in
the proof of Theorem~\ref{thm:width}(i), the numbers $\alpha_\pi$ and
$\beta_\pi$ of every piece of the skeleton are carried along the
construction as centered forms over $\Phi$, by the rules of
Lemma~\ref{lem:scaling} and the arithmetic of
Lemma~\ref{lem:centered}. The ranges of
$(\beta_\pi + 2a\ell_\pi)/\alpha_\pi$ over $\Phi$ are compared; on
every $\Phi$ the largest belongs to the last piece of ball 19, every
other piece staying below $86.376$, and
over the twenty-three pieces these ranges lie in
$[86.869064, 86.886898]$ m$^{2}$/s$^{2}$, whose square root is the
stated interval. Over the box of half-width $10^{-90}$ about
$\varphi^{*}$, at 1200 bits, the same computation gives
$V_c(\varphi^{*})^{2} = 86.8860775531756$ and
$V_c(\varphi^{*}) = 9.32127016844677$, to fifteen figures. The last
assertion is Lemma~\ref{lem:skeleton} together with
Lemma~\ref{lem:scaling}.

(ii) As in the proof of Theorem~\ref{thm:width}(i), with the speed
range $[9.321315, V_0^{*} - 2.3]$ covered by ten intervals, of lengths
$0.05$ (five of them), $0.025$, $0.003$, $0.0006$ and $0.00008$ m/s,
and at the bottom the interval $[9.321315, 9.3213201]$, each certified
over the same twenty-three pieces of
$J_\delta$. The lower end $9.321315$ exceeds every value of
$V_c$ on $J_\delta$ by at least $8 \times 10^{-7}$ m/s, and at the
lower end ball 19 enters its pocket at about $2.5 \times 10^{-4}$ m/s.

(iii) As in the proof of Theorem~\ref{thm:ribbon}, with the single
speed interval $[9.3212701695, 9.321315]$ and the box of directions of
half-width $10^{-90}$; the upper part is covered by (ii). At the lower
end of this strip ball 19 enters its pocket at about $10^{-5}$ m/s.
Since $9.3212701695 - V_c(\varphi^{*}) < 1.1 \times 10^{-9}$, the edge
at $\varphi^{*}$ is located to that accuracy.
\end{proof}

At the speed $V_c(\varphi)$ ball 19, moving along the skeleton, would
arrive at the pocket at $(0, 0)$ with speed zero. This is the sense in
which the ribbon ends where ball 19 comes to rest at the lip of its
pocket. Whether $N(c^{*}, x) = 21$ can hold for some
slower stroke, through a different list of events, is not decided
here. The sharpness of the edge is a feature of Model M, in which a
ball is pocketed at any positive speed.

\begin{remark}[The frictionless model]
\label{rem:frictionless}
If $a = 0$, Lemma~\ref{lem:scaling} reads $v^{2} = \alpha_\pi V_0^{2}$:
all speeds are proportional to $V_0$. The evolution from
$(V_0, \varphi)$ is then the evolution from $(1, \varphi)$ run at
$V_0$ times the rate, every ball being at time $t / V_0$ where the
latter is at time $t$. Collisions and cushion contacts occur at the
same places and in the same order, so $(V_0, \varphi)$ realizes a list
of events, or pockets all balls, exactly when $(1, \varphi)$ does. The
sets $F$ and $\{N = 21\}$ are then products of $(0, V_{\max}]$ with
sets of directions, there is no lower edge, and the whole structure of
the successful strokes lies in the direction coordinate. With $a > 0$
the speed enters in two ways. One is the stops, which
Lemma~\ref{lem:scaling} reduces to the single threshold
$V_c(\varphi)$. The other is timing: the time $(v_0 - v_1)/a$ spent on
a piece is not a homogeneous function of $V_0$, so the order of the
events of different balls, and the distances between balls that pass
near each other, change with $V_0$ and could in principle break the
list of events at some speed above $V_c(\varphi)$.
Theorem~\ref{thm:edge} shows that this does not happen in the witness:
across the whole window of directions, the list of events is lost, as
the speed decreases, within $10^{-3}$ m/s of the point where ball 19
comes to rest at the lip of its pocket, and nowhere above it.
\end{remark}

The structure of Proposition~\ref{prop:semialg} applies to the sets
of the last two sections.

\begin{corollary}[Algebraic edges]
\label{cor:algebraic}
The set $F$ of Section~\ref{sec:thin} is semialgebraic and defined over
$\mathbb{Q}$, and its connected component $F^{*}$ is semialgebraic and
defined over the field of real algebraic numbers. Hence the two
endpoints of the window
$J$ of Theorem~\ref{thm:width} are algebraic points of the unit
circle, the lower edge $V_0 = V_c(\varphi)$ of Theorem~\ref{thm:edge}
is an arc of a real algebraic curve, and the Lebesgue measure of
$F^{*}$ is a period in the sense of Kontsevich and Zagier
\cite{kontsevich2001}.
\end{corollary}

\begin{proof}
The set $F$ prescribes nothing for the cue ball after its last event
of $\Sigma$. From that event on it meets no other ball, since the
object balls have no events outside $\Sigma$, so it is in free motion,
along a path of length at most $V_0^{2}/2a \le 720$ m. It meets the
cushions $x = R$ and $x = L - R$ at most $720/(L - 2R) + 1$ times and
the other two at most $720/(W - 2R) + 1$ times, and then it stops or is
pocketed. So $F$ is the union of the sets $F_{\Sigma'}(c^{*})$ over the
finitely many lists $\Sigma'$ that append to $\Sigma$ events of the cue
ball alone, namely cushion contacts and then its stop or its pocketing,
and is semialgebraic over $\mathbb{Q}$ by
Proposition~\ref{prop:semialg}.
The connected components of a semialgebraic set defined over
$\mathbb{Q}$ are semialgebraic sets defined over the real closure of
$\mathbb{Q}$, the field of real algebraic numbers \cite{bcr1998}, so
$F^{*}$ is defined over that field, and so is the projection $J$ of
$F^{*}$ to $S^{1}$, by the Tarski--Seidenberg theorem. A semialgebraic
subset of $S^{1}$ defined over the real algebraic numbers is a finite
union of points and arcs whose endpoints have algebraic coordinates, as
one sees by the rational parametrization of the circle. The construction of the skeleton uses
rational operations and square roots of the coordinates of
$(\cos\varphi, \sin\varphi)$, so the numbers $\alpha_\pi$,
$\beta_\pi$ and $\ell_\pi$ of Lemma~\ref{lem:scaling}, and
$V_c(\varphi)$, are algebraic functions of that point. Finally, the
measure of $F^{*}$ is the integral of the form
$dV_0 \wedge (\xi\, d\eta - \eta\, d\xi)$, with
$(\xi, \eta) = (\cos\varphi, \sin\varphi)$, over a bounded
semialgebraic set defined over the real algebraic numbers. Kontsevich
and Zagier define periods through domains and integrands with rational
coefficients and note that algebraic coefficients give the same class
\cite{kontsevich2001}, so this integral is a period.
\end{proof}

The certificates of Theorem~\ref{thm:width} thus locate two algebraic
points of the unit circle to within $10^{-47}$ radians. We know nothing
about their degrees.

\section{The opening position}
\label{sec:break}

\begin{conjecture}
\label{conj:break}
There is a stroke $x \in \OO$, regular for $c_0^{\varepsilon}$, with
$N(c_0^{\varepsilon}, x) = 21$.
\end{conjecture}

By Lemma~\ref{lem:continuity}, the conjecture implies that the set of
such strokes has positive measure. It rests on the richness of the
cascade that follows the first impact into the pack, and on the
absence of any evident obstruction. Energy is no obstruction, since a
stroke at $V_{\max}$ carries enough kinetic energy to roll a ball
$V_{\max}^{2} / 2a = 720$ m. The question is lopsided, in the following
precise sense.

\begin{proposition}[Certified search]
\label{prop:search}
There is a procedure that, run on $c_0^{\varepsilon}$, halts only after
certifying, through Corollary~\ref{cor:chain}, a box of strokes on
which $N(c_0^{\varepsilon}, \cdot) = 21$, and that halts whenever
Conjecture~\ref{conj:break} holds.
\end{proposition}

\begin{proof}
For $m = 1, 2, \dots$, run the certification of
Section~\ref{sec:cert} with $2^{m}$ bits of precision, in the proposals
as in the enclosures, on every dyadic
box of side $2^{-m}$ contained in $\OO$, with time brackets of width
$2^{-m/2}$ about the proposed event times, allowing at most $2^{m}$
events and $2^{m}$ bisection steps per positivity check, and halt
when some box is certified to end with all twenty-one object balls
pocketed. A halt is a proof, by Corollary~\ref{cor:chain} and the
argument of Theorem~\ref{thm:box}.

Suppose $x^{*}$ is regular for $c_0^{\varepsilon}$ with $N = 21$. Along
the evolution from $x^{*}$ each event clock has a simple zero, every
other clock is bounded away from zero up to that zero, except fresh
clocks whose quotients are bounded away from zero, and every stop clock
is positive. By continuity these bounds persist, with smaller
positive constants, for all strokes in some neighborhood $B$ of
$x^{*}$, and the event sequence is constant on $B$. For large $m$ some
dyadic box of side $2^{-m}$ lies in $B$. As the width of a box and the
rounding error tend to zero, the interval enclosures of the states,
clock coefficients and polynomial values converge to the true values,
uniformly on $B$, because they are built from Lipschitz operations.
Hence for $m$ large every strict inequality required by
Lemma~\ref{lem:step} is verified on that box, the floating-point
proposals follow the true event sequence, and the procedure halts.
\end{proof}

No comparable procedure is known that halts when the conjecture is
false. Sampling cannot provide one. No feasible number of random
strokes can distinguish an astronomically small positive measure from
zero. A disproof by exhaustive certified search would have to rule out
ribbons as thin as the one of Section~\ref{sec:thin}, and
nothing known excludes such ribbons from the break. A negative answer
seems to require a
structural argument, and we know of none.

\section{A second clearance with a wide window}
\label{sec:cluster}

The window of $c^{*}$ is narrow because twenty thin cuts act in series.
A configuration with fuller cuts and short legs behaves differently,
and most of all where a ball strikes a resting neighbor across a
gap of a millimeter or two: the struck ball then leaves along the line
of centers, which is nearly fixed by the resting geometry. We exhibit
such a configuration. It was
found by a search in the double-precision simulator
(Section~\ref{sec:search}); the claims below rest on the certificate
alone.

The configuration $c_{\circ}$ (\texttt{construction/cluster\_21.json},
Figure~\ref{fig:cluster}, Table~\ref{tab:cluster}) has twenty-one object balls in a few tight
groups in the left half of the table. Along the certified stroke $(V_0^{\circ}, \varphi^{\circ})$
with $V_0^{\circ} = 12$ m/s the cue ball strikes twelve of them
directly, at cut angles between $45^{\circ}$ and $75^{\circ}$ and over
legs between $1.5$ mm and $0.54$ m. Of the other nine, four are struck
by a neighbor that has itself just been struck and has rolled $1.5$
mm, and five by a ball that has rolled between $0.09$ and $0.46$ m;
the cut angles of these nine collisions lie between $19.5^{\circ}$ and
$46^{\circ}$. Every collision is of a moving ball onto a resting one,
and no object ball is more than three collisions away from the cue
ball. The list $\Sigma_{\circ}$ of the events up to the pocketing of
the last object ball has $57$ events, namely $21$ ball-ball
collisions, $15$ cushion rebounds and $21$ pocketings, and the last
ball drops at $5.05$ s. Let $F_{\circ}$ be the set of strokes that
realize $\Sigma_{\circ}$.

\begin{figure}[ht]
\centering
\includegraphics[width=\linewidth]{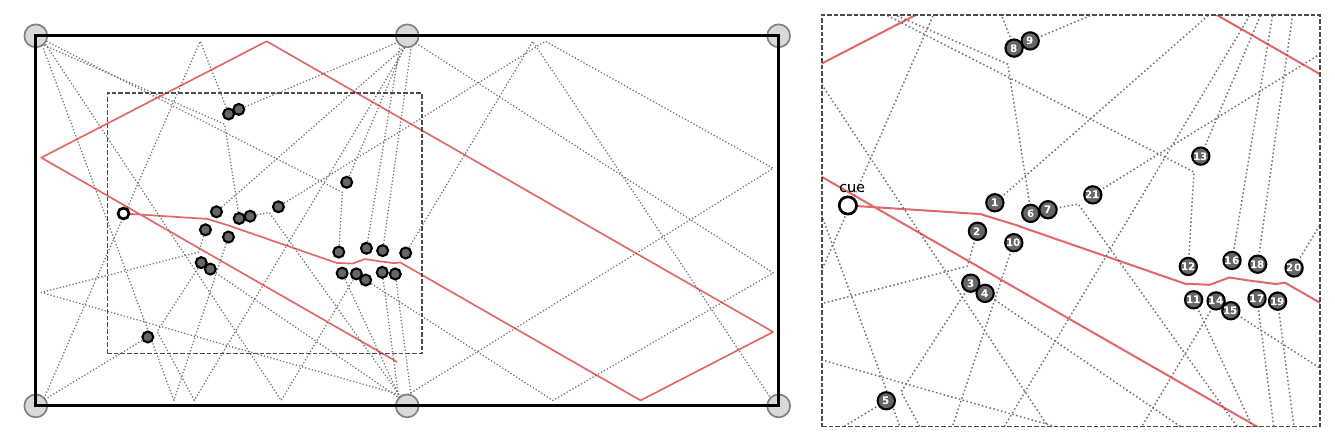}
\caption{The cluster $c_{\circ}$ and the certified stroke
$(V_0^{\circ}, \varphi^{\circ})$. Left: the whole table, drawn as in
Figure~\ref{fig:witness}. Right: the region inside the dashed
rectangle, enlarged, with the object balls numbered in the order of
\texttt{construction/cluster\_21.json}.}
\label{fig:cluster}
\end{figure}

\paragraph{Taylor models.}
The centered forms of Section~\ref{sec:thin} fail on $c_{\circ}$. The
true slope of every quantity with respect to $\varphi$ stays of order
$1$ to $10$, but a slope computed in interval arithmetic loses the
cancellation that keeps it small: its width grows by about
$1/(2R\cos\theta)$ at every collision, so that pieces narrower than
$4 \times 10^{-12}$ would be needed. The remedy is an enclosure of
higher order, a univariate Taylor model in the sense of Berz and
Makino \cite{berz1998}. Fix a piece
$\Phi = [\varphi_c - h, \varphi_c + h]$ and a degree $N \ge 1$, and
write $t = \varphi - \varphi_c$ and $T = [-h, h]$. A \emph{Taylor
model} of degree $N$ for a function $q$ on $\Phi$ is a pair $(P, I)$
of a polynomial $P = \sum_{k=0}^{N} a_k t^{k}$ with interval
coefficients and an interval $I$ containing $0$, the remainder, with
\[
q(\varphi) \in \sum_{k=0}^{N} a_k\, t^{k} + I
\qquad \text{for all } \varphi \in \Phi .
\]
A centered form is the case $N = 1$ and $I = \{0\}$. Let $T^{k}$ be the
range of $t^{k}$ on $T$, that is $[-h^{k}, h^{k}]$ for odd $k$ and
$[0, h^{k}]$ for even $k$, and let $\rho(P, I)$ be the range of the
model, the interval obtained by evaluating $P$ on $T$ in interval
arithmetic and adding $I$.

\begin{lemma}[Taylor models]
\label{lem:taylor}
Let $(P_1, I_1)$ and $(P_2, I_2)$ be Taylor models of degree $N$ for
$q_1$ and $q_2$ on $\Phi$, with coefficients $a_{1,k}$ and $a_{2,k}$.
\begin{enumerate}[label=(\roman*), leftmargin=2.5em]
\item $(P_1 \pm P_2,\, I_1 \pm I_2)$ is a Taylor model for
$q_1 \pm q_2$.
\item Let $c_k = \sum_{i+j=k} a_{1,i}\, a_{2,j}$ for $0 \le k \le 2N$.
Then
\[
\Bigl(\, \sum_{k=0}^{N} c_k\, t^{k},\ \
\sum_{k=N+1}^{2N} c_k\, T^{k} + I_1\, \rho(P_2, 0) + I_2\, \rho(P_1, 0)
+ I_1 I_2 \Bigr)
\]
is a Taylor model for $q_1 q_2$.
\item Let $X = \rho(P_1, I_1)$, and let $f(x) = 1/x$ with
$0 \notin X$, or $f(x) = \sqrt{x}$ with $X > 0$. Let $m$ be a point of
$a_{1,0}$, let $d = q_1 - m$ with the model obtained from $(P_1, I_1)$
by subtracting $m$ from $a_{1,0}$, and let $\Delta$ be the range of
that model. If $(P, I)$ is the model obtained for
$\sum_{k=0}^{N} f^{(k)}(m)\, d^{k} / k!$ by (i) and (ii), then
\[
\Bigl(\, P,\ \ I + \frac{f^{(N+1)}(X)}{(N+1)!}\, \Delta^{N+1} \Bigr)
\]
is a Taylor model for $f(q_1)$. Here $f^{(N+1)}(x)/(N+1)!$ is
$(-1)^{N+1} x^{-(N+2)}$ for the reciprocal and
$\binom{1/2}{N+1}\, x^{-N-1/2}$ for the square root.
\item The Taylor polynomials of degree $N$ of $\cos$ and $\sin$ at
$\varphi_c$ are Taylor models for $\cos\varphi$ and $\sin\varphi$ when
each is given the remainder
$[-h^{N+1}, h^{N+1}] / (N+1)!$.
\end{enumerate}
In each case the new remainder contains $0$.
\end{lemma}

\begin{proof}
For $\varphi \in \Phi$ write $q_i(\varphi) = p_i + r_i$ with
$p_i = \sum_{k \le N} a_{i,k}'\, t^{k}$, $a_{i,k}' \in a_{i,k}$ and
$r_i \in I_i$. Then (i) is immediate, and (ii) is the identity
\[
q_1 q_2 = \sum_{k=0}^{2N} c_k'\, t^{k} + r_1 p_2 + r_2 p_1 + r_1 r_2,
\qquad c_k' = \sum_{i+j=k} a_{1,i}'\, a_{2,j}' \in c_k ,
\]
with $t^{k} \in T^{k}$ and $p_i \in \rho(P_i, 0)$. For (iii), Taylor's
formula with the Lagrange remainder gives
\[
f(q_1) = \sum_{k=0}^{N} \frac{f^{(k)}(m)}{k!}\, d^{k}
+ \frac{f^{(N+1)}(\xi)}{(N+1)!}\, d^{N+1}
\]
for some $\xi$ between $m$ and $q_1(\varphi)$. Both lie in the interval
$X$, the second by the definition of a model and the first because
$m \in a_{1,0}$, $0 \in T$ and $0 \in I_1$. Hence $\xi \in X$, on which
$f$ is smooth, and $d \in \Delta$. Part (iv) is Taylor's formula for
$\cos$ and $\sin$, whose derivatives are bounded by $1$. Finally
$0 \in T^{k}$ for $k \ge 1$, $0 \in I_i$ and $0 \in \Delta$, so each
new remainder contains $0$.
\end{proof}

The coefficients of a Taylor model are computed from thin data and stay
thin. The inflation of interval arithmetic acts on the remainder, which
is of order $h^{N+1}$. With $N = 6$ the enclosures of the skeleton of
$c_{\circ}$ on pieces of width $1.5 \times 10^{-5}$ are sharp to about
$2 \times 10^{-10}$ m.

\begin{theorem}[A wide window]
\label{thm:cluster}
Every stroke with
\[
\varphi^{\circ} - 2.5 \times 10^{-5} \le \varphi
\le \varphi^{\circ} + 9.15 \times 10^{-4}
\quad\text{and}\quad
V_0^{\circ} - 0.01 \le V_0 \le V_0^{\circ}
\]
lies in $F_{\circ}$ and pockets all twenty-one object balls of
$c_{\circ}$. The window of directions of $c_{\circ}$ is therefore at
least $9.4 \times 10^{-4}$ radians wide, and the set
$\{\, x \in \OO : N(c_{\circ}, x) = 21 \,\}$ has Lebesgue measure at
least $9.4 \times 10^{-6}$.
\end{theorem}

\begin{proof}
The rectangle is the union of ten rectangles, one for each of the
speed intervals
$[V_0^{\circ} - 10^{-3} k,\ V_0^{\circ} - 10^{-3} (k - 1)]$ with
$1 \le k \le 10$, and each of these is covered by $64$ pieces of
directions of width $1.469 \times 10^{-5}$. On each piece the
hypotheses of
Proposition~\ref{prop:skeleton}, read for $c_{\circ}$ and
$\Sigma_{\circ}$, are verified as in the proof of
Theorem~\ref{thm:width}(i), with Taylor models of degree $N = 6$ and
Lemma~\ref{lem:taylor} in place of centered forms and
Lemma~\ref{lem:centered}. The geometric conditions are verified on the
ranges of the models, and the timing conditions, the positivity of
speeds in (S1), and (S3), (S4) and (S5), exactly as before. Every stroke in
the rectangle therefore realizes $\Sigma_{\circ}$. The slowest pocket
entry over the rectangle is at $0.248$ m/s or more. After the last
pocketing only the cue ball is on the table, and its remaining motion
has finitely many events as in the proof of Theorem~\ref{thm:box}, so
$N(c_{\circ}, x) = 21$. The measure bound is the area of the rectangle.
(\texttt{proof/taylor\_models.py}, \texttt{proof/certificates/cluster/}.)
\end{proof}

\begin{remark}
\label{rem:cluster}
The rectangle of Theorem~\ref{thm:cluster} is not the whole window of
$c_{\circ}$, and its edges are not located: within about
$1.5 \times 10^{-5}$ radians of either end the timing check cannot
separate one pair of balls, because on four of its legs the cue ball,
on its way to one ball, passes the ball it strikes after that one with
a clearance between
$0.09$ and $0.63$ mm. For the same reason the check fails on a speed
interval of depth $10^{-2}$ m/s taken whole, so the rectangle is
certified as ten intervals of depth $10^{-3}$, and nothing is claimed
for lower speeds. The
configuration is thus forgiving in aim and unforgiving in the placement
of its balls.
The edges are nevertheless algebraic. The coordinates of $c_{\circ}$
are rational, so the argument of Corollary~\ref{cor:algebraic} applies
without change: $F_{\circ}$ is semialgebraic over $\mathbb{Q}$, and the
endpoints of the window of directions of its component through
$(V_0^{\circ}, \varphi^{\circ})$ are algebraic points of the unit
circle.
\end{remark}

\section{How the configurations were found}
\label{sec:search}

This section is a record of experiment. It proves nothing, and no
theorem depends on it. Both configurations were found by simulation,
and what failed on the way explains their shape.

\paragraph{The chain.}
Each object ball is placed by the ``ghost ball'' rule run in reverse.
A point is chosen on the cue ball's known path and a pocket is chosen,
and the ball is placed at distance $2R$ beyond the point on the line
toward the pocket, so that the struck ball leaves toward the pocket.
The first constructors computed the whole chain in closed form and
placed all the balls at once. None of them produced a chain that the
simulator reproduced, because the predicted path of the cue ball
drifted from the simulated one as cushion contacts accumulated. The
constructor that succeeded (\texttt{construction/witness\_bv.py}) grows
the chain one collision at a time inside the simulator, and integrates
the whole stroke again after every placement. It prefers thin cuts,
because a thin cut costs the cue ball little speed. In $c^{*}$ the
first twenty cuts lie between $72^{\circ}$ and $84^{\circ}$, the cue
ball keeps at least $95\%$ of its speed at each of them, and only the
last collision is nearly head-on.

The first complete layout was built and checked in double precision,
and was reported in version~1 of this preprint. It is not a witness of
Model M. Replayed in 80-digit arithmetic, the evolution follows the
design for nine events, the next two occur in the opposite order, at
the 12th one object ball strikes another, and four balls are pocketed. Two features of the double-precision
engine are responsible, and neither belongs to the model: rounding, and
a separation of $1$ nm that the engine inserts between two bodies after
each contact, so that the contact is not detected twice. Every thin cut
amplifies both (Section~\ref{sec:thin}). The same constructor, with
every placement verified in 80-digit arithmetic
(\texttt{construction/witness\_bv\_exact.py}), produced $c^{*}$. For a
chain of thin cuts, then, a replay in double precision says little
about Model M, and this is why every claim of this paper rests on a
certificate.

\paragraph{Trees.}
A window of $10^{-42}$ radians raises the question whether a wide one
exists. The first attempt replaced the chain by a tree
(\texttt{construction/witness\_tree.py}). Every moving ball strikes a
few resting balls with moderate cuts over short legs and then drops
into a pocket itself, so that no ball depends on more than a few
earlier collisions. Each committed step was verified on a stencil of
nearby strokes. The search closed at six to ten balls. The obstacle was
the last ball struck by each carrier, after which the deflected carrier
must itself find a pocket. In Model M a ball aimed at a pocket point
is accepted only if its line makes at most $16.5^{\circ}$ with the
diagonal of a corner pocket, or at least $28.5^{\circ}$ with the
cushion of a middle pocket, and few placements satisfy this condition
and the stencil at once.

\paragraph{The cluster.}
The search that succeeded (\texttt{construction/witness\_cluster.py})
keeps the tree and shortens its inner legs to between $1.5$ and $3$ mm,
after the frozen combinations of trick-shot play. Only the last ball
struck by a carrier may sit farther along its path, up to $0.6$ m,
since its position is the freedom that aims the deflected carrier at a
pocket. The cue ball strikes its balls with cuts of $30^{\circ}$ or
more, so that it keeps its speed, and every other moving ball strikes
one or two balls and then rolls into a pocket. Each step was verified on a stencil
of strokes at $\pm 0.05$ m/s and $\pm 2 \times 10^{-5}$ radians. Runs with
legs of the cue ball of at most $0.25$ m stopped short of twenty-one
balls. A run with legs of at most $0.6$ m placed twenty, and a
completion step added the twenty-first as a further ball struck by one
already placed. The arguments of that run and the code of the last step
were not kept, and the script as it stands in the repository stops at
fourteen to sixteen balls. The result of the original run is
$c_{\circ}$, which Theorem~\ref{thm:cluster} certifies whatever its
origin.

Two expectations were corrected on the way. The design assumed that a
near-frozen collision damps an error by the factor $2R\cos\theta / g$,
where $g$ is the gap. This was too optimistic, since the lateral offset
of the striker at contact grows with the distance $2R + g$ between the
centers and not with the gap. What the simulation showed is weaker and
sufficient: the derivatives of all positions with respect to the
direction of the stroke stay of order $1$ to $10$. Secondly, the
constructor has a blind spot that the certificate exposed. It checks
the clearance of every pair of balls that is not meant to collide, but
not of a pair that is meant to collide later, and on four of its legs
the cue ball, on its way to one ball, passes the ball it strikes after
that one at $0.09$,
$0.12$, $0.41$ and $0.63$ mm. These near-contacts, and not the aim,
limit the rectangle of Theorem~\ref{thm:cluster}
(Remark~\ref{rem:cluster}). A search that checks them is the natural
next experiment (Problem (P4)).

\section{Open problems}
\label{sec:open}

The successful strokes of Theorems~\ref{thm:box} to~\ref{thm:edge}
form a set of positive measure that no one could hit; the certified
rectangle of Theorem~\ref{thm:cluster} is wide in direction but, as
certified, only $10^{-2}$ m/s deep in speed. We collect the
questions that the results leave open.

\begin{enumerate}[label=(P\arabic*), leftmargin=2.8em]
\item \emph{The break.} Decide Conjecture~\ref{conj:break}. The
procedure of Proposition~\ref{prop:search} halts if the conjecture
holds with a regular witness. Nothing is known that would halt if it
fails, and successful strokes from
the break, if they exist, may be as thin as the ribbon of $c^{*}$.

\item \emph{Finiteness of events.} For a fixed admissible
configuration $c$, is the number of events of the evolution from
$(c, x)$ bounded over $x \in \OO$, or at least finite for every $x$?
By Lemma~\ref{lem:duration} infinitely many events would have to
accumulate before time $V_0 / a$. Chains of inelastic particles in one
dimension do show such accumulations, known as inelastic collapse
\cite{bernu1990,mcnamara1992}, and nothing in Model M visibly excludes
them. A bound on the number of events would make the set of successful
strokes semialgebraic, with finitely many connected components, by
Proposition~\ref{prop:semialg}.

\item \emph{The whole success set of $c^{*}$.} Theorems~\ref{thm:width}
and~\ref{thm:edge} describe $F^{*}$, the component of the strokes that
realize the witness's list of events. Is
$\{\, x \in \OO : N(c^{*}, x) = 21 \,\}$ connected, and how far does it
extend beyond $F^{*}$? Just outside the window of directions, ball 21
rebounds from a cushion inside the jaws of its pocket and may still
drop (Remark~\ref{rem:eventlist}); just below $V_c(\varphi)$, ball 19
comes to rest at the lip of its pocket (Section~\ref{sec:edge}).
Whether some other list of events pockets all twenty-one balls in
either case is not known. The endpoints of the window $J$ are
algebraic points of the unit circle (Corollary~\ref{cor:algebraic});
nothing is known about their degrees.

\item \emph{Wide windows.} The window of $c^{*}$ is
$9.5 \times 10^{-43}$ radians wide because twenty thin cuts act in
series, and that of $c_{\circ}$ is at least $9.4 \times 10^{-4}$
radians wide (Theorem~\ref{thm:cluster}). What is the supremum, over
admissible configurations, of the width of the window of directions,
or of the measure of the set of successful strokes? The cluster is
forgiving in aim only: on four of its legs the cue ball clears a ball
by less than a millimeter (Remark~\ref{rem:cluster}). Is there a
configuration cleared by a set of strokes that is wide in direction
and in speed, and that is still cleared by nearby strokes when each of
its balls is moved by up to a millimeter? Finally, where do the edges
of the window of $c_{\circ}$ lie, and which events decide them?

\item \emph{Variations of the model.} Without deceleration the set of
successful strokes is a product of $(0, V_{\max}]$ with a set of
directions (Remark~\ref{rem:frictionless}), and the lower edge of
Theorem~\ref{thm:edge} exists because Model M pockets a ball at any
positive speed. Which other features of the model matter? Does a
configuration with a total clearance exist when a pocket accepts a
ball only above a minimum entry speed, or only within a limited range
of directions, as real pockets do, and for every choice of the
constants $e$, $e_c$, $a$ and $\rho$ in reasonable ranges?
\end{enumerate}

\appendix

\section{The witness}
\label{app:witness}

\IfFileExists{witness_layout.pdf}{%
\begin{figure}[ht]
\centering
\includegraphics[width=0.95\linewidth]{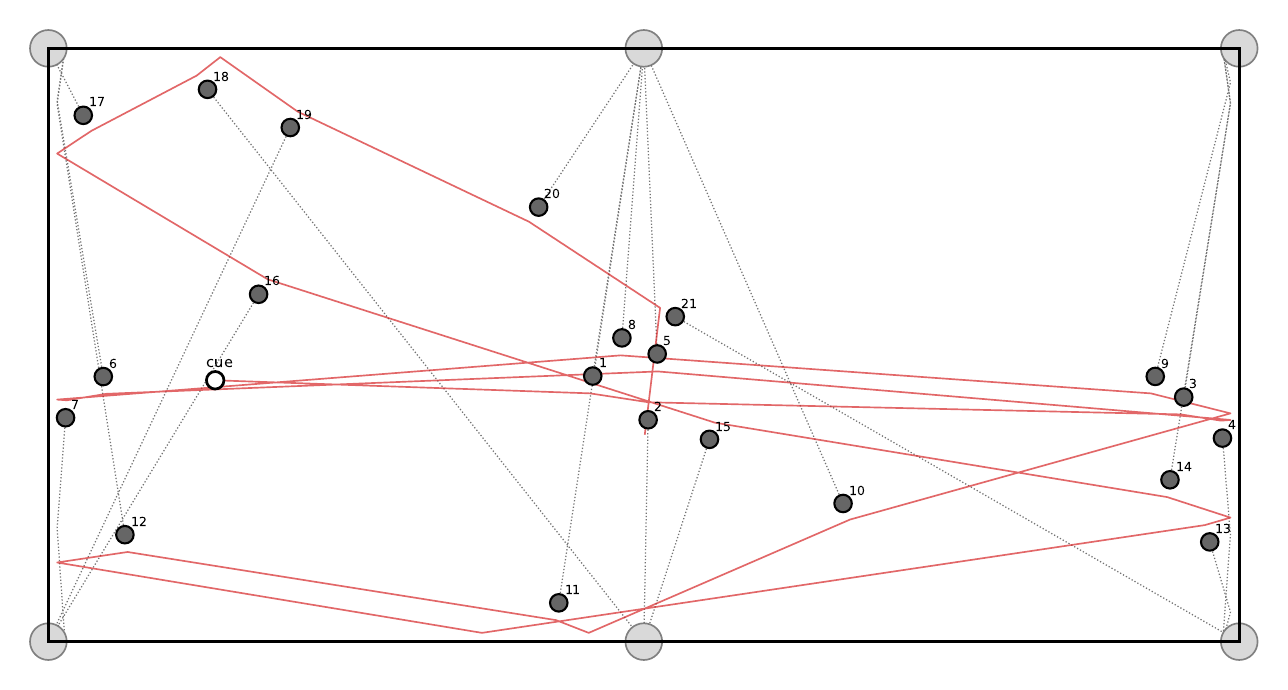}
\caption{The witness. The cue ball (open) strikes all twenty-one object
balls (dark) in one stroke, its path shown in red. Each struck ball's
path to its pocket is dotted. Pocket capture disks are shaded.}
\label{fig:witness}
\end{figure}}{}

The witness (Figure~\ref{fig:witness}) was constructed by the greedy
search described in
Section~\ref{sec:search}, with every placement verified in 80-digit
arithmetic. The layout reported in version~1 of this preprint, built in
double precision, is not a witness of Model M.

\paragraph{Model instance.} $L=3.569$, $W=1.778$, $R=0.02625$, $a=0.10$, $e=0.95$, $e_c=0.90$, $\rho=0.055$ (SI, exact decimals).
\paragraph{The stroke.} Cue ball at $(0.5000,\,0.7833)$, $V_0^{*}\approx 11.9$ m/s, $\varphi^{*}\approx -0.034915$ rad.
\begin{center}\begin{tabular}{rrrl}
\toprule
\# & $x$ (mm) & $y$ (mm) & pocket \\
\midrule
1 & 1631.5 & 795.9 & M-t \\
2 & 1797.4 & 664.8 & M-b \\
3 & 3402.6 & 732.8 & C-tr \\
4 & 3518.8 & 609.9 & C-br \\
5 & 1824.6 & 862.2 & M-t \\
6 & 164.6 & 794.4 & C-tl \\
7 & 51.5 & 670.9 & C-bl \\
8 & 1718.9 & 910.3 & M-t \\
9 & 3317.3 & 794.8 & C-tr \\
10 & 2381.7 & 414.0 & M-t \\
11 & 1529.7 & 116.3 & M-t \\
12 & 229.4 & 320.5 & C-tl \\
13 & 3480.6 & 298.9 & C-br \\
14 & 3361.4 & 485.0 & C-tr \\
15 & 1981.4 & 606.3 & M-b \\
16 & 630.1 & 1041.0 & C-bl \\
17 & 104.6 & 1577.3 & C-tl \\
18 & 477.0 & 1655.2 & M-b \\
19 & 725.1 & 1540.8 & C-bl \\
20 & 1469.4 & 1302.2 & M-t \\
21 & 1878.8 & 974.0 & C-br \\
\bottomrule
\end{tabular}\end{center}
Pockets are labeled C for a corner and M for the middle of a long side, with b or t for the side $y = 0$ or $y = W$ and, for corners, l or r for the end $x = 0$ or $x = L$. The certificate files number the object balls from $0$, so ball $k$ of this paper is \texttt{b}$(k-1)$ there.
\paragraph{Exact values.} The table is rounded. The configuration and stroke of Theorems~\ref{thm:box} and~\ref{thm:ribbon} are the binary64 numbers stored in \texttt{proof/witness\_21.json}, taken as exact rationals; for example $V_0^{*} = 11.9000000000000003552713678800500929355621337890625$, not the decimal $11.9$.

Table~\ref{tab:events} lists the events of $\Sigma$, and
Table~\ref{tab:cluster} gives the configuration $c_{\circ}$ of
Section~\ref{sec:cluster}.

\begin{table}[ht]
\centering
\caption{The event list $\Sigma$ of the witness, in the order of the events at the
stroke $x^{*}$, with their times there rounded to $10^{-4}$ s. Balls are numbered
as in the table of the witness. The cushions are named by their side: left
$x = R$, right $x = L - R$, bottom $y = R$, top $y = W - R$.}
\label{tab:events}
\footnotesize
\begin{tabular}{@{}rrl@{\hspace{1.2em}}rrl@{\hspace{1.2em}}rrl@{}}
\toprule
\# & $t$ (s) & event & \# & $t$ (s) & event & \# & $t$ (s) & event \\
\midrule
1 & 0.0945 & cue strikes 1 & 22 & 1.1745 & cue strikes 10 & 43 & 2.6596 & 10 pocketed, M-t \\
2 & 0.1095 & cue strikes 2 & 23 & 1.2878 & cue, bottom cushion & 44 & 2.7882 & cue, left cushion \\
3 & 0.2460 & cue strikes 3 & 24 & 1.3011 & 6, left cushion & 45 & 2.8140 & cue strikes 17 \\
4 & 0.2565 & cue strikes 4 & 25 & 1.3020 & cue strikes 11 & 46 & 2.8890 & cue strikes 18 \\
5 & 0.2590 & cue, right cushion & 26 & 1.4123 & 6 pocketed, C-tl & 47 & 2.9083 & cue, top cushion \\
6 & 0.3670 & 4, right cushion & 27 & 1.4489 & 8 pocketed, M-t & 48 & 2.9715 & cue strikes 19 \\
7 & 0.4290 & cue strikes 5 & 28 & 1.4820 & cue strikes 12 & 49 & 3.1455 & cue strikes 20 \\
8 & 0.4784 & 4 pocketed, C-br & 29 & 1.5131 & cue, left cushion & 50 & 3.1543 & 15 pocketed, M-b \\
9 & 0.5174 & 2 pocketed, M-b & 30 & 1.6150 & 9, right cushion & 51 & 3.1734 & 17, left cushion \\
10 & 0.5940 & cue strikes 6 & 31 & 1.6634 & 9 pocketed, C-tr & 52 & 3.1802 & 17 pocketed, C-tl \\
11 & 0.6060 & cue strikes 7 & 32 & 1.7220 & cue, bottom cushion & 53 & 3.2535 & cue strikes 21 \\
12 & 0.6090 & cue, left cushion & 33 & 2.0790 & cue strikes 13 & 54 & 3.6544 & 16 pocketed, C-bl \\
13 & 0.7501 & 7, left cushion & 34 & 2.0922 & 12, left cushion & 55 & 3.6573 & 14, right cushion \\
14 & 0.7906 & 1 pocketed, M-t & 35 & 2.0922 & cue, right cushion & 56 & 3.7039 & 21 pocketed, C-br \\
15 & 0.8040 & cue strikes 8 & 36 & 2.1285 & cue strikes 14 & 57 & 3.8601 & 14 pocketed, C-tr \\
16 & 0.8382 & 3, right cushion & 37 & 2.1575 & 12 pocketed, C-tl & 58 & 4.0736 & 20 pocketed, M-t \\
17 & 0.8889 & 7 pocketed, C-bl & 38 & 2.3474 & 13, right cushion & 59 & 5.5996 & 19, left cushion \\
18 & 0.9304 & 3 pocketed, C-tr & 39 & 2.3696 & 11 pocketed, M-t & 60 & 5.6221 & 19 pocketed, C-bl \\
19 & 0.9900 & cue strikes 9 & 40 & 2.3820 & cue strikes 15 & 61 & 5.7902 & 18 pocketed, M-b \\
20 & 1.0192 & cue, right cushion & 41 & 2.4137 & 13 pocketed, C-br &  &  &  \\
21 & 1.1310 & 5 pocketed, M-t & 42 & 2.6460 & cue strikes 16 &  &  &  \\
\bottomrule
\end{tabular}
\end{table}

\begin{table}[ht]
\centering
\caption{The cluster $c_{\circ}$. The cue ball is at $(421.389,\, 924.383)$ mm, and the
stroke has $V_0^{\circ} = 12$ m/s and $\varphi^{\circ} \approx -0.065022$ rad. Pocket
labels are as in the table of the witness. The table is rounded; the exact
configuration and stroke are the binary64 numbers stored in
\texttt{construction/cluster\_21.json}.}
\label{tab:cluster}
\footnotesize
\begin{tabular}{@{}rrrll@{\hspace{2em}}rrrll@{}}
\toprule
\# & $x$ (mm) & $y$ (mm) & struck by & pocket & \# & $x$ (mm) & $y$ (mm) & struck by & pocket \\
\midrule
1 & 867.926 & 932.509 & cue & M-t & 12 & 1455.625 & 738.981 & cue & C-tl \\
2 & 815.097 & 846.016 & cue & M-b & 13 & 1493.576 & 1074.316 & 12 & M-t \\
3 & 794.676 & 688.301 & 2 & M-t & 14 & 1540.200 & 634.352 & cue & C-tl \\
4 & 838.354 & 657.252 & 3 & M-b & 15 & 1584.900 & 604.556 & 14 & M-t \\
5 & 538.134 & 330.917 & 3 & C-bl & 16 & 1588.536 & 757.201 & cue & M-t \\
6 & 976.933 & 900.383 & cue & C-tl & 17 & 1664.453 & 641.658 & cue & M-b \\
7 & 1029.521 & 911.356 & 6 & M-b & 18 & 1666.418 & 745.893 & cue & M-t \\
8 & 926.293 & 1402.490 & 6 & C-bl & 19 & 1727.043 & 633.700 & cue & M-b \\
9 & 975.204 & 1424.297 & 8 & M-t & 20 & 1777.056 & 734.450 & cue & C-br \\
10 & 925.512 & 811.270 & cue & C-tl & 21 & 1165.579 & 956.536 & 7 & M-b \\
11 & 1471.801 & 637.763 & cue & M-b &  &  &  &  &  \\
\bottomrule
\end{tabular}
\end{table}

\section{Reproduction}
\label{sec:repro}

All code, data and certificate outputs are available at
\url{https://github.com/avnerkantor/one-shot-twenty-one-balls}.
\nolinkurl{proof/certify_mpmath.py} and \nolinkurl{proof/certify_arb.py}
are the two certifications of Theorem~\ref{thm:box}, and
\nolinkurl{proof/certify_strip.py} is the certification of
Theorem~\ref{thm:ribbon}. For Theorem~\ref{thm:width},
\nolinkurl{proof/certify_window.py} certifies the rectangle of (i) in
centered forms, one speed interval at a time, and
\nolinkurl{proof/certify_miss.py} certifies the two misses of (ii)
and brackets the edges by bisection. For Theorem~\ref{thm:edge},
\nolinkurl{proof/certify_lower.py} encloses the threshold $V_c$ over
the window, \nolinkurl{proof/certify_window.py} certifies the ten
further speed intervals of (ii), and \nolinkurl{proof/certify_strip.py}
the strip of (iii). The first-order analysis of Section~\ref{sec:thin} is reproduced by \nolinkurl{analysis/first_order.py}, which composes the transfer maps along the certified skeleton and compares the result with the slopes of the centered forms; its output is
\nolinkurl{analysis/first_order_output.txt}. The outputs of the
certifications are stored in
\nolinkurl{proof/certificates/}, those of Theorem~\ref{thm:width} in
\nolinkurl{proof/certificates/window/} and
\nolinkurl{proof/certificates/miss/}, and those of
Theorem~\ref{thm:edge} in \nolinkurl{proof/certificates/lower/},
\nolinkurl{proof/certificates/window/} and
\nolinkurl{proof/certificates/strip/}. The witness is
\nolinkurl{proof/witness_21.json}, built by
\nolinkurl{construction/witness_bv_exact.py} with seed 7 using the
80-digit engine \nolinkurl{model/model_sim_mp.py}. The certifications
need the mpmath and python-flint packages. The certificates in the
repository were produced with Python 3.13.5, mpmath 1.3.0 with its pure
Python backend, and python-flint 0.9.0, which includes FLINT 3.6.0 and
with it Arb.

\begin{table}[!htb]
\centering
\caption{The certificates. Bits is the working precision. An empty
cell means that the certificate establishes no pocket-entry speed.}
\label{tab:certs}
\small
\begin{tabular}{@{}llrrll@{}}
\toprule
Statement & Enclosure & Bits & Pieces & Slowest entry (m/s) & Time (s) \\
\midrule
Theorem~\ref{thm:box} & boxes & 1200 & 1 & $0.44$ & 85 \\
Theorem~\ref{thm:ribbon} & boxes & 400 & 48 & $0.149$ & 322 \\
Theorem~\ref{thm:width}(i) & centered forms & 400 & 1104 & $0.149$ & 2619 \\
Theorem~\ref{thm:width}(ii) & boxes & 1200 & 2 & & 0.5 \\
Theorem~\ref{thm:edge}(i) & centered forms & 400 & 23 & & 31 \\
Theorem~\ref{thm:edge}(i), at $\varphi^{*}$ & centered forms & 1200 & 1 & & 1 \\
Theorem~\ref{thm:edge}(ii) & centered forms & 400 & 230 & $2.4 \times 10^{-4}$ & 512 \\
Theorem~\ref{thm:edge}(iii) & boxes & 400 & 1 & $9.1 \times 10^{-6}$ & 8 \\
Theorem~\ref{thm:cluster} & Taylor models, degree 6 & 400 & $10 \times 64$ & $0.248$ & $1.2 \times 10^{5}$ \\
\bottomrule
\end{tabular}
\end{table}

As a check from outside, \nolinkurl{proof/crosscheck_arb.py} runs the
Arb program at 52 single strokes: at the corners and on a grid of the
certified rectangles it certifies that the stroke realizes the list of
events, and at strokes just outside the window and below the critical
speed that it does not. \nolinkurl{proof/test_taylor_models.py} tests
the Taylor models for containment. Neither is part of a proof.

Table~\ref{tab:certs} lists the size of each certificate. A piece is
one evaluation of the whole skeleton on a set of strokes. The slowest
entry is the least pocket-entry speed that the certificate
establishes, rounded down; it measures the margin in the positivity of speeds
required by (S1), and it is small
in Theorem~\ref{thm:edge} because the rectangles there reach down to
the edge of the ribbon. The times are wall-clock times on one desktop
machine, summed over the runs of a certificate. Several runs were made
at a time, which lengthens each of them, so the times indicate scale
only. For Theorem~\ref{thm:cluster} there are ten runs of $64$ pieces
each, made four at a time.

\section*{Acknowledgements}

I thank Vered Rom-Kedar, whose work on billiards, described in a
Weizmann Institute article I read some twenty years ago, first raised
this question for me. I used Claude (Anthropic) to assist in
restructuring and copyediting this manuscript, including contributing
supporting arguments and citations for some technical points and
writing much of the code for the computer proofs. The research
question, the overall approach, and the final text are my
responsibility.

\end{document}